\documentclass[fleqn,a4paper]{article}

\usepackage[utf8]{inputenc}
\usepackage[T2A,T1]{fontenc}
\usepackage{hyperref}
\usepackage{graphicx}
\usepackage[russian,english]{babel}

\usepackage{amsmath}
\usepackage{amssymb}
\usepackage{amsthm}						% BUGFIX: amsthm needs to appear _after_ amsmath 
\usepackage{enumerate}
\usepackage{caption}
\usepackage{subcaption}

\usepackage{lmodern}					% ...smoothing out aaaaaall the edges

\usepackage{xfrac}						% for \sfrac

\usepackage[top=24mm,left=30mm,right=33mm,heightrounded,head=14.5pt]{geometry}
\usepackage[numbers,sort&compress,longnamesfirst,sectionbib]{natbib}

\usepackage{authblk}
\usepackage[bottom]{footmisc}

\newtheorem{theorem}{Theorem}
\newtheorem{lemma}[theorem]{Lemma}
\newtheorem{corollary}[theorem]{Corollary}

\usepackage{timetravel}
\setCounters{theorem}

\makeatletter
\DeclareRobustCommand{\bfseries}{%
  \not@math@alphabet\bfseries\mathbf
  \fontseries\bfdefault\selectfont
  \boldmath
}
\makeatother

\newcommand*{\Hyp}{\mathcal{H}}
\newcommand*{\Gyp}{\mathcal{G}}
\newcommand*{\B}{\mathcal B}
\newcommand*{\p}{\mathbf p}
\newcommand*{\nwspace}{\hspace*{.1em}} 

\DeclareMathOperator{\Or}{O}
\DeclareMathOperator{\littleO}{o}
\DeclareMathOperator{\littleOm}{\omega}
\DeclareMathOperator{\Om}{\Omega}
\DeclareMathOperator{\Th}{\Theta}
\DeclareMathOperator{\ld}{ld}
\DeclareMathOperator{\Hb}{H}
\DeclareMathOperator{\Dkl}{D}
\DeclareMathOperator{\Ev}{E}

\DeclareMathOperator{\Pb}{P}
\DeclareMathOperator{\Bin}{Bin}
\DeclareMathOperator{\poly}{poly}
\DeclareMathOperator*{\argmin}{arg\,min}
\DeclareMathOperator*{\argmax}{arg\,max}

\providecommand{\ignore}[1]{}

\begin{document}

\title{%
	{\Large\textbf{The Minimization of Random Hypergraphs}\thanks{%
		An extended abstract was presented at the
		28th European Symposium on Algorithms (ESA 2020)~\cite{Blaesius20MinimizationESA}.
		}
	}\\[.5em]
}

\author[1]{Thomas Bl{\"a}sius}
\author[2]{Tobias Friedrich}
\author[2]{Martin Schirneck}

\affil[1]{Karlsruhe Institute of Technology, Karlsruhe, Germany}
\affil[2]{Hasso Plattner Institute, University of Potsdam, Germany

	\emph{\texttt{firstname.lastname@\{kit.edu\nwspace{|}\nwspace{}hpi.de\}}}}

\date{}

\maketitle

\begin{abstract}
	\noindent
	We investigate the maximum-entropy model $\B_{n,m,p}$ for random $n$-vertex,
	$m$-edge multi-hypergraphs with expected edge size $pn$.
	We show that the expected size of the minimization of $\B_{n,m,p}$,
	i.e., the number of its inclusion-wise minimal edges,
	undergoes a phase transition with respect to $m$.
	If $m$ is at most $1/(1-p)^{(1-p)n}$, then the minimization is of size $\Th(m)$.
	Beyond that point, for $\alpha$ such that $m = 1/(1-p)^{\alpha n}$ 
	and $\Hb$ being the entropy function, it is
	\begin{equation*}
	 	\Th(1) \cdot \min\!\left(1, \, \frac{1}{(\alpha\,{-}\,(1-p))
	 		\nwspace \sqrt{(1\,{-}\,\alpha) \nwspace n}}\right) 
	 		\cdot 2^{(\Hb(\alpha) + (1-\alpha) \log_2 p)  \nwspace n}.
	\end{equation*} 

	\noindent
	This implies that the maximum expected size over all $m$ is $\Th((1+p)^n/\sqrt{n})$.
	
	Our structural findings have algorithmic implications
	for minimizing an input hypergraph,
	which in turn has applications in the profiling of relational databases
	as well as for the Orthogonal Vectors problem studied in fine-grained complexity.
	
	The main technical tool is an improvement of the Chernoff--Hoeffding inequality,
	which we make tight up to constant factors.
	We show that for a binomial variable $X \sim \Bin(n,p)$
	and real number $0 < x \le p$, it holds that
	\begin{equation*}
		\Pb[X \le xn] = \Th(1) \cdot \min\!\left(1, \, \frac{1}{(p-x) \nwspace \sqrt{xn}}\right) \cdot 
			2^{-\!\Dkl(x \,{\|}\, p) \nwspace n},
	\end{equation*}
	
	\noindent
	where $\Dkl$ denotes the Kullback--Leibler divergence between Bernoulli distributions.
	The result remains true if $x$ depends on $n$ as long as it is bounded away from $0$.
\end{abstract}

\ignore{%Pure LaTeX abstract.
	We investigate the maximum-entropy model $\mathcal{B}_{n,m,p}$ for random $n$-vertex,
	$m$-edge multi-hypergraphs with expected edge size $pn$.
	We show that the expected size of the minimization of $\mathcal{B}_{n,m,p}$,
	i.e., the number of its inclusion-wise minimal edges,
	undergoes a phase transition with respect to $m$.
	If $m$ is at most $1/(1-p)^{(1-p)n}$, then the minimization is of size $\Theta(m)$.
	Beyond that point, for $\alpha$ such that $m = 1/(1-p)^{\alpha n}$ 
	and $\mathrm{H}$ being the entropy function, it is
	$\Theta(1) \cdot \min\!\left(1, \, \frac{1}{(\alpha\,{-}\,(1-p))
	 	\sqrt{(1\,{-}\,\alpha) n}}\right) 
	 	\cdot 2^{(\mathrm{H}(\alpha) + (1-\alpha) \log_2 p) n}.$
	This implies that the maximum expected size over all $m$ is $\Theta((1+p)^n/\sqrt{n})$.
	
	Our structural findings have algorithmic implications
	for minimizing an input hypergraph,
	which in turn has applications in the profiling of relational databases
	as well as for the Orthogonal Vectors problem studied in fine-grained complexity.
	
	The main technical tool is an improvement of the Chernoff--Hoeffding inequality,
	which we make tight up to constant factors.
	We show that for a binomial variable $X \sim \mathrm{Bin}(n,p)$
	and real number $0 < x \le p$, it holds that
	$\mathrm{P}[X \le xn] = 
		\Theta(1) \cdot \min\!\left(1, \, \frac{1}{(p-x) \sqrt{xn}}\right) \cdot 
		2^{-\!\mathrm{D}(x \,{\|}\, p) n}$,
	where $\mathrm{D}$ denotes the Kullback--Leibler divergence between Bernoulli distributions.
	The result remains true if $x$ depends on $n$ as long as it is bounded away from $0$.
}

\section{Introduction}
\label{sec:intro}

A plethora of work has been dedicated to the analysis of random graphs.
Random hypergraphs, however, received much less attention.
For many types of data, hypergraphs provide a much more natural model.
This is especially true if the data has a hierarchical structure
or reflects interactions between groups of entities.
In non-uniform hypergraphs, where edges can have different numbers of vertices,
a phenomenon occurs that is unknown to graphs:
an edge may be contained in another, with multiple edges even forming chains of inclusion.
We are often only interested in the endpoints of those chains,
namely, the collections of inclusion-wise minimal or maximal edges.
This is the \emph{minimization} or \emph{maximization} of the hypergraph, respectively.

We investigate the maximum-entropy model $\B_{n,m,p}$
for random multi-hypergraphs with $n$ vertices and $m$ edges and expected edge size $pn$
for a constant sampling probability $p$.
In other words, out of all probability distributions on hypergraphs
that result in an expected edge size of $pn$,
$\B_{n,m,p}$ is the one of maximum entropy.
We are interested in the expected size of the minimization/maximization of $\B_{n,m,p}$,
that is, the expected number of its minimal/maximal edges.
Our results are phrased in terms of the minimization,
but replacing the probability $p$ with $1-p$ immediately transfers them to the maximization.
We show that the size of the minimization undergoes a phase transition 
with respect to the number of edges $m$
with the point of transition at $m^* = 1/(1-p)^{(1-p)n}$.
While the number of edges is still small, a constant fraction of them is minimal
and the minimization grows linearly in the total sample sizes.
For $m > m^*$, the size of the minimization is instead governed
by the entropy function of the exponent $\alpha$ such that $m = 1/(1-p)^{\alpha n}$
(see Theorem~\ref{thm:properties_Bnpk} for a precise statement).
We show that the ratio of minimal edges goes down exponentially as $m$ increases.
This seems counter-intuitive at first as it decouples the behavior of the minimization
from the growth of the underlying hypergraph.
We further determine that the maximum expected number of
minimal edges over all $m$ is of order $\Th((1+p)^n/\sqrt{n})$,
the maximum is attained at $m = 1/(1-p)^{\frac{n}{1+p}}$.

Our results establish a close connection between the size of the minimization
and the binomial distribution.\footnote{%
	The notation $\B_{n,m,p}$ is mnemonic of the binomial distribution 
	emerging in the sampling process.
}
Theorem~\ref{thm:binomial_characterization} (see below)
characterizes the number of minimal edges in terms of the total number of edges $m$
and the likelihood that a binomial variable deviates from its expectation.
The main tool in our analysis is the Chernoff--Hoeffding theorem bounding
the tail of the distribution function via the Kullback--Leibler divergence
from information theory.
However, the existing inequalities are not sharp enough
to derive tight statements on the expected size of the minimization.
So far, there is an $\Or(\sqrt{n})$ gap between the best-known upper and lower estimates.
In this work, we improve both sides such that they match up to constant factors.
The improvement transfers directly to the weaker, but more practical,
versions of Chernoff bounds widely used in the literature.
While our findings concentrates on the asymptotics,
we also give an explicit interval for the leading constants involved
and our proofs indicate how large $n$ must be to be applicable.
This makes the result useful also in a non-asymptotic setting
and in areas beyond the treatment of random hypergraphs.
We illustrate this by giving an alternative proof that
the binomial distribution does not vanish outside 
of a constant number of standard deviations from its mean.
Our proof avoids the Berry--Esseen inequality, that is, normal approximation.
We also improve the rate of convergence in Cramér's theorem.

Regarding our main topic of hypergraphs,
our structural insights have algorithmic implications
for the task of actually computing the minimization $\min(\Hyp)$ of an input hypergraph $\Hyp$.
We give two examples from fine-grained complexity as well as data profiling.
There is reason to believe that there exists no minimization algorithm
running in time \mbox{$\Or(m^{2-\varepsilon}) \cdot \poly(n)$} for any $\varepsilon > 0$
on \mbox{$m$-edge}, \mbox{$n$-vertex} hypergraphs.
The argument uses the Sperner Family problem,
which is to decide whether $\Hyp$ contains two edges such that one is contained in the other,
that is, whether $|\min(\Hyp)| < |\Hyp|$.
The latter is equivalent to the, arguably more prominent, Orthogonal Vectors problem
studied in fine-grained complexity.
A truly subquadratic algorithm for the minimization would
thus falsify the Orthogonal Vectors Conjecture 
and in turn the Strong Exponential Time Hypothesis (SETH),
resulting in a major breakthrough in Boolean satisfiability.
On the other hand, partitioning the edges of the hypergraph $\Hyp$ by their cardinality
and processing them in order of increasing cardinality
gives an $\Or(mn \nwspace |\min(\Hyp)|)$-time algorithm,
which is $\Or(m^2 \nwspace n)$ in the worst case.
When looking at the average-case complexity for the $\B_{n,m,p}$ distribution,
we get a run time of $\Or(mn \Ev[ \nwspace |\min(\B_{n,m,p})| \nwspace])$.
Our results therefore show that the algorithm is subquadratic on average
for all $m$ beyond the phase transition,
it is even linear for $m$ larger than $1/(1-p)^n$.

There is also a connection to the profiling of relational databases.
Data scientists regularly need to compile and output a comprehensive list of metadata,
like unique column combinations, functional dependencies, or, most general, denial constraints.
These multi-column dependencies can all be described as the minimal hitting sets
of certain hypergraphs created from comparing pairs of rows in the database and recording 
the sets of attributes in which they differ.
Computing these difference sets one by one
generates an incoming stream of seemingly random subsets.
Filtering the inclusion-wise minimal ones from the stream does not affect the solution,
but can greatly reduce the number of sets
and the complexity of the resulting hitting set instance.
Minimizing the input is therefore a standard preprocessing technique in data profiling.
In real-world databases, there are often fewer minimal difference sets
than rows in the database, let alone pairs thereof.
Therefore, the space needed to store the sets
usually makes up only a small fraction of the original input size.
The upper bounds given in Theorems~\ref{thm:properties_Bnpk}
provide a theoretical explanation for this observation.
We show that only a few difference sets can be expected to be minimal
and their number even shrinks as the database grows larger.
Conversely, the difference sets and the corresponding 
multi-column dependencies are mutually dual,
which allows to recover the minimized input from the
the collection of all solutions.
In this sense, the matching lower bounds in Theorems~\ref{thm:properties_Bnpk} can be seen
as the smallest amount of data any enumeration algorithm needs to process
in order to correctly cover all dependencies.

\vspace*{-.75em}
\subparagraph*{Related Work.}

Erd{\H o}s--R{\'e}nyi graphs $\Gyp_{n,m}$~\cite{Erdos59RandomGraphs}
and Gilbert graphs $\Gyp_{n,p}$~\cite{Gilbert59RandomGraphs}
are arguably the most discussed random graph models in the literature.
We refer the reader to the monograph by Bollob{\'a}s~\cite{Bollobas01RandomGraphs} for an overview.
A majority of the work on these models concentrates on various
phase transitions with respect to the number of edges $m$
or the sample probability $p$, respectively.
This intensive treatment is fueled by the appealing property
that Erd{\H o}s--R{\' e}nyi graphs are ``maximally random''
in that they do not assume anything but the number of vertices and edges.
More formally, among all probability distributions on graphs with $n$ vertices and $m$ edges,
$\Gyp_{n,m}$ is the unique distribution of maximum entropy.
The same holds for $\Gyp_{n,p}$
under the constraint that the expected number of edges is $p \binom{n}{2}$, see~\cite{Anand09Entropy}.

The intuition of being maximally random is captured by the Shannon entropy,
which is the central concept in information theory~\cite{Cover2006InformationTheory,Shannon48Communication}.
A discrete stochastic system
described by the probability distribution $(p_i)_i$
has a (binary) entropy of $H((p_i)_i) = - \sum_i \nwspace p_i \log_2 p_i$.
The self-information of a single state with probability $p$ is ${-}\log_2 p$,
the entropy is the expected information of the whole system.
It is a measure of surprisal or how ``spread out'' the distribution is.
Originally stemming from thermodynamics~\cite{Lieb91PhysicsReport},
the versatility of this definition is key to the successful application
of information theory to fields as diverse as cryptography~\cite{Bruen04Cryptography},
machine learning~\cite{Grandvalet06EntropyRegularization},
quantum computing~\cite{Nielsen10QuantumComputing}, 
and of course network analysis~\cite{Newman10Networks}.
Jaynes' \emph{principle of maximum entropy} states that out of an ensemble of probability distributions
that all describe the observed phenomena equally well,
the one of maximum entropy is to be preferred in order to minimize any outside bias~\cite{Kesavan09Jaynes,Jaynes57MaximumEntropyI,Jaynes57MaximumEntropyII}.
Properties of the maximum-entropy model
then pertain to the average system matching the observations.
In the context of random graphs,
it is used to define so-called null models~\cite{Zweig14Literacy}.
After certain graph statistics observed in real-world networks are fixed,
one chooses the maximum-entropy distribution that meets these constraints.
By comparing the original network with a ``typical'' graph drawn from the null model, 
one can infer whether other properties are correlated with the constraints.
This method was made rigorous by Park and Newman~\cite{Park04StatisticalMechanics}
building on earlier work in general statistics.
Prescribing the exact or expected number of edges
leads to the $\Gyp_{n,m}$ or $\Gyp_{n,p}$ distributions, respectively.
The configuration model fixes the whole degree sequence of the graph~\cite{Bollobas80ConfigurationModel}
and the soft configuration model relaxes these constraints to hold in expectation~\cite{Bianconi07NetworkEnsembles,Garlaschelli08UnbiasedInformation}.

Many early attempts to transfer the concept of null models to hypergraphs have been only indirect
in that they have studied hypergraphs via their clique expansion~\cite{Newman01Collaboration}
or as bipartite graphs~\cite{Saracco15Bipartite}.
This is unsatisfactory since the projections alter relevant observables,
like node degrees or the number of triangles.
Only recently, Chodrow generalized the configuration model directly to multi-hypergraphs~\cite{Chodrow20ConfigurationComplexNetworks},
which has subsequently been refined by Arafat~et~al.~\cite{Arafat20HypergraphsPrescribedDegreeSequenceArXiv}.
There also seems to be not much work on hypergraph models
that can be cast into the maximum-entropy framework without being intentionally designed as such,
a notable exception is the work by Schmidt-Pruzan and Shamir~\cite{SchmidtPruzan85Component}.
They fix the exact (respectively, expected) edge sequence
such that the largest edge has cardinality $\Or(\log n)$
and show a ``double jump'' phase transition in the size of the largest connected component.
Most of the recent literature on random hypergraphs concentrates on the $k$-uniform model
where every edge has exactly $k$ vertices~\cite{Karonski02PhaseTransition,Behrisch010GiantComponent,Behrisch14LocalLimit}
or, equivalently, on random binary matrices with $k$ $1$s per column~\cite{Cooper19RankJournal}.
Our model neither prescribes the exact cardinalities of the edges 
nor a bound on their maximum size,
instead it only requires that the \emph{expected} edge size is $pn$.

Probably closest to our work is a string of articles by Demetrovics et al.~\cite{Demetrovics98RandomDatabases}
as well as Katona~\cite{Katona2012CorrelatedData,Katona13FunctionalConnection}.
They investigate random databases and
connect the R{\' e}nyi entropy of \emph{order~2} of the logarithmic number of rows
with the probability that certain unique column combinations or functional dependencies hold.
In contrast, we connect the Shannon entropy of the logarithmic number of \emph{pairs} of rows
(quantity $\alpha$, see Section~\ref{sec:main_thm})
with the expected number of minimal difference sets.
Unique column combinations, functional dependencies,
and their generalization to denial constraints are dual to the difference sets,
one are the hitting sets of the other~\cite{Abedjan2018DataProfilingBook,Bertossi11DatabaseRepairingBook,
Blaesius16IPEC,Froese16DistinctVectors}.
Also, the Shannon entropy is equal to the R{\' e}nyi entropy of \emph{order~1}~\cite{Cover2006InformationTheory}.
In this sense, we complement the result by Demetrovics~et~al.\
and show that the duality also pertains to the order of entropy.
Furthermore, it has often been observed in practice 
that the collection of minimal difference sets of real-world databases
is much smaller than the original instance size~\cite{Blaesius2019Efficiently,Papenbrock15SevenAlgorithms,
Bleifuss17DenialConstraints,Livshits20ApproximatePVLDB},
our findings provide a theoretical explanation for this phenomenon.

The minimization of hypergraphs also occurs in fine-grained complexity
in form of the Sperner Family problem.
It is subquadratically equivalent to the Orthogonal Vectors problem~\cite{Borassi2016IntoTheSquare,Gao2018CompletenessFO},
which in turn admits a fine-grained reduction from CNF-Satisfiability~\cite{Williams2005SETHtoOV}.
Any truly subquadratic algorithm for computing the minimization of a hypergraph
\emph{in the worst case} would mark a major breakthrough in satisfiability.
Very recently, ideas from fine-grained complexity have been extended to
to the average case~\cite{Ball17AverageCaseFineGrainedHardness,KaneWilliams18OVBranchingPrograms,
Dalirrooyfard20FineGrainedAverageCaseArXiv}.
We show that a simple algorithm for Sperner Family is subquadratic \emph{on average}
on hypergraphs with expected edge size $pn$.

The analysis of random (hyper-)graphs naturally builds
on tools from combinatorics and probability theory.
Conversely, it has always helped to advance those fields by improving on known techniques~\cite{BergeHypergraphs,Bollobas01RandomGraphs,Hofstad16RandomGraphs}.
The binomial distribution is used all throughout science to describe complex systems
emerging from the overlapping effects of many independent choices.
Its typical behavior is well understood,
as described in the central limit theorem
and the strong law of large numbers.
Bounding its tails, however, remains the subject of ongoing research.
The resulting concentration inequalities
play a significant role in the theory of large deviations~\cite{Deuschel84LargeDeviations,Dembo10LargeDeviations},
the analysis of randomized algorithms and search heuristics~\cite{Mitzenmacher17Probablility,Dubhashi09ConcentrationOfMeasure,Doerr20Evolutionary},
as well as computational learning~\cite{Kearns94ComputationalLearning},
to name a few examples.
In this work, we sharpen the tail inequalities of the Chernoff--Hoeffding theorem~\cite{Hoeffding63Inequalities}.
We prominently use a result by Klar~\cite{Klar00TailProbabilities}
on the relation between the distribution function and the probability mass function.
Some refined inequalities have been known before.
By Cram{\'e}r's theorem~\cite{Cramer38SurUn},
Chernoff--Hoeffding is asymptotically tight up to subexponential factors
and the gap has subsequently been reduced to $\Or(\sqrt{n})$, see~\cite{Ash90InformationTheory}.
We close it down to a constant.
There also exist comparatively tight bounds 
based on the normal limit of the binomial distribution.
Contributions by Prokhorov~\cite{Prokhorov53Asymptotic} and later Slud~\cite{Slud77BinomialLaw}
founded major lines of research in that direction.
However, we avoid this approach and give a purely combinatorial argument
since the normal approximation cannot be expressed in terms of elementary functions. 
Also, it tends to place unnecessary restrictions on the sampling probability $p$
for non-asymptotic results.

\vspace*{-.75em}
\subparagraph*{Outline.}

In the next section, we introduce the hypergraph model and state our results in full detail.
We review some notation and general concepts in Section~\ref{sec:prelims}.
Section~\ref{sec:Chernoff--Hoeffding} is dedicated to the Chernoff--Hoeffding theorem.
Section~\ref{sec:distinct_sets_and_minimality} adds further technical contributions
as the foundation of the subsequent proofs.
The main theorems on the size of the minimization
are then proven in Section~\ref{sec:proof_of_main_thm}.
Section~\ref{sec:conclusion} concludes the work.

\section{Model and Main Results}
\label{sec:main_thm}

Fix a probability $p$ and positive integers $n$ and $m$.
The random multi-hypergraph $\B_{n,m,p}$ is generated
by independently sampling $m$ (not necessarily distinct) subsets of $[n]$.
Each set contains any vertex $v \in [n]$
with probability $p$ independently of all other choices.\footnote{%
	In the context of Boolean functions,
	this is called the 
	\emph{$p$-biased distribution}~\cite{ODonnell14BooleanFunctions}.
}

We quickly argue that it is indeed the maximum-entropy model.
Besides the size of the universe $n$ and the number of edges $m$,
the only constraint is the expected edge size $pn$.
The independence bound on the entropy reads as follows.
Let $X_1$ to $X_m$ be random variables with joint distribution
$\Pb_{X_1, \dots, X_m}$ and marginals $\Pb_{X_j}$.
Then, their entropies observe the inequality
$H(\Pb_{X_1, \dots, X_m}) \le \sum_{j = 1}^m H(\Pb_{X_j})$,
with equality holding if and only if the $X_j$ are independent~\cite{Cover2006InformationTheory}.
This implies that the edges need to be sampled independently in order to maximize the entropy,
the same is true for the vertices inside one edge.
Finally, the fact that setting the sampling probability of the vertices to be all equal
indeed gives the maximum entropy under a given mean set size
was proven by Harremo{\"e}s~\cite{Harremoes01MaximumEntropy}.

We are interested in the expected number of inclusion-wise minimal edges in $\B_{n,m,p}$,
denoted by $\Ev[ \nwspace |\min(\B_{n,m,p})| \nwspace ]$.
We describe the asymptotic behavior of this expectation with respect to $n$.
In more detail, we view the number of edges $m = m(n)$ as a function of $n$
and bound the univariate asymptotics of $\Ev[ \nwspace |\min(\B_{n,m,p})| \nwspace ]$ 
with respect to $n$ for any choice of $m$.
The sampling probability $p$ is considered to be a constant throughout.

We show first that the expected size of the minimization can be described precisely
in terms of $m$ and the binomial distribution $\Bin(n,p)$.
To state our result in full detail,
we define
\begin{equation*}
	\alpha = \log_{\frac{1}{(1-p)^n}} m =  - \frac{\log_{1-p} m}{n}.
\end{equation*}

\noindent
The quantity $\alpha$ is a non-negative function of $p$, $n$, and $m$,
it is well-defined for all $0 < p < 1$ and $n,\,m \ge 1$.
Asymptotically in $n$, it is of order $\Th((\log m)/n)$.
If $p$ and $n$ are fixed, choosing a value for $\alpha$ determines $m$
since we can rewrite $m$ as $1/(1-p)^{\alpha n}$.

\begin{figure}[tb]
  \begin{subfigure}[t]{0.49\textwidth}
    \centering
    \includegraphics[page=2]{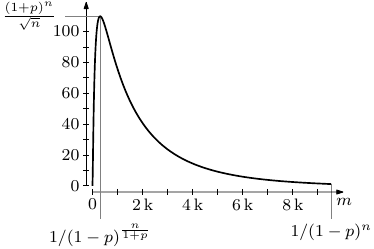}
    \caption{The expected size of the minimization as a function of $m$
      for $n \,{=}\, 10$ and $p \,{=}\, 0.6$ in the information-theoretic regime.
	  The vertical line at
      $m = 1/(1-p)^{\frac{n}{1+p}}$ indicates the position of the maximum
      (Lemma~\ref{lem:maximum_Bnpk}).  For $m > 1/(1-p)^{n}$, the
      size goes to~$1$.
      The linear bound for $m \le 1/(1-p)^{(1-p)n}$ is not shown as it
      is too close to $0$.
    }
    \label{fig:plots-by-m}
  \end{subfigure}\hfill
  \begin{subfigure}[t]{0.49\textwidth}
    \centering
    \includegraphics[page=3]{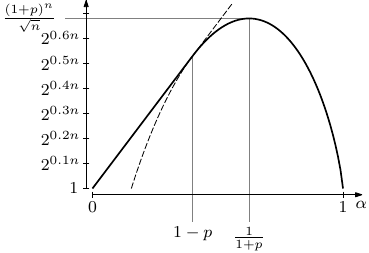}
    \caption{The expected size of the minimization as a function of
      $\alpha$ for $p = 0.6$ (the plot is independent of $n$).
      The vertical line at $\alpha = 1-p$  indicates the phase transition
      between the linear and the information-theoretic regime.
      The respective bounds are continued as dashed lines into the other regime.
      The vertical line at $\alpha = 1/(1+p)$ 
      indicates the position of the maximum (Lemma~\ref{lem:maximum_Bnpk}). 
  	}
    \label{fig:plots-by-alpha}
  \end{subfigure}
  \caption{Illustration of Theorem~\ref{thm:properties_Bnpk} showing 
    the expected size of the minimization of a random hypergraph
    depending on the number of edges $m$~(\subref{fig:plots-by-m}) and
    on $\alpha$~(\subref{fig:plots-by-alpha}).  As $\alpha$ grows
    logarithmically in $m$, (\subref{fig:plots-by-alpha}) shows the
    same plot as (\subref{fig:plots-by-m}) but with both axes being
    logarithmic.}
  \label{fig:plots}
\end{figure}

\begin{theorem}
\label{thm:binomial_characterization}
	Let  $n$, $m$ be positive integers and $p$ a probability.
	If $p = 0$ or $p = 1$, then $|\min(\B_{n,m,p})| = 1$ holds deterministically.
	For $0 < p < 1$, let $X \sim \Bin(n,p)$ be a binomially distributed random variable.
	\begin{enumerate}[\textup{(}$i$\textup{)}]
		\item\label{clause:char_thm_almostall}
			For any $\varepsilon > 0$ and all $m \le 1/(1-p)^{(1-\varepsilon) n}$,
			i.e., all $0 \le \alpha \le 1 - \varepsilon$,
			\begin{equation*}
				\Ev[ \nwspace |\min(\B_{n,m,p})| \nwspace ] 
					= \Th(m) \cdot \Pb[ \nwspace X \le (1-\alpha) n \nwspace].
			\end{equation*}
		\item\label{clause:char_thm_trivial_case}
			There exists a $c>0$
			such that for all $m \ge 1/(1-p)^{n+ c \nwspace \ln n}$, i.e., all 
			$\alpha \ge  1 + \frac{c \nwspace \ln n}{n}$,
			\begin{equation*}
				\Ev[ \nwspace |\min(\B_{n,m,p})| \nwspace ] = 1 + \littleO(1).
			\end{equation*}
	\end{enumerate}
\end{theorem}

\noindent
The asymptotic estimate in the first statement is tight up to constants,
the second statement is even tight up to lower-order terms.
The constants hidden in the big-O notation are universal in the sense
that they, of course, do not depend on $m$ or $n$, 
but also not on $\alpha$,
which describes the relation between the former two.
We note that the constants may depend on $p$ and $\varepsilon$.
There is a gap in the theorem at $m = 1/(1-p)^{n}$,
which can, however, be made arbitrarily small.
Let $C = 1/(1-p)$, then Statement~($\ref{clause:char_thm_almostall}$) 
holds if $m \le (C - \delta)^n$ for any constant $\delta > 0$ and
Statement~($\ref{clause:char_thm_trivial_case}$) takes over at $m \ge (C + \littleO(1))^n$.
Unlike what one might expect,
we show (in Lemma~\ref{lem:cannot_extend})
that Theorem~\ref{thm:binomial_characterization}~($\ref{clause:char_thm_almostall}$)
cannot be extended to the case where $\alpha$ converges to $1$.

The characterization of the size of the minimization via the binomial distribution
allows us to give estimates in information-theoretic terms,
namely, using the Shannon entropy of $\alpha$.
This reveals a phase transition at $m^* = 1/(1-p)^{(1-p)n}$, i.e., for $\alpha = 1-p$.
Let $\Hb(x) = - x \log_2 x - (1-x) \log_2 (1-x)$
denote the binary entropy function.

\pagebreak

\wormhole{thm:properties_Bnpk}
\begin{theorem}
\label{thm:properties_Bnpk}
	Let  $n$, $m$ be positive integers and $0 < p < 1$ a non-trivial probability.
	\begin{enumerate}[\textup{(}$i$\textup{)}]
		\item\label{clause:main_thm_linear_range}
			If $m \le 1/(1-p)^{(1-p)n}$,
			then $\Ev[ \nwspace |\min(\B_{n,m,p})| \nwspace] = \Th(m)$.\vspace*{.5em}
		\item\label{clause:main_thm_information-theoretic_range}
			For any $\varepsilon > 0$ and all 
			$1/(1-p)^{(1-p)n} \le m \le 1/(1-p)^{(1-\varepsilon) n}$,
			i.e., all $1-p \le \alpha \le 1 - \varepsilon$,
			\begin{align*}
				\Ev[ \nwspace |\min(\B_{n,m,p})| \nwspace ]
					&= \Th(1) \cdot \min\!\left(1, 
						\, \frac{1}{(\alpha\,{-}\,(1-p)) \nwspace \sqrt{(1\,{-}\,\alpha) n}}\right) 
	 					\cdot 2^{(\Hb(\alpha) + (1-\alpha) \log_2 p)  \nwspace n}\\[.25em]
					&= \Th(1) \cdot \min\!\left(1, 
						\, \frac{1}{(\alpha\,{-}\,(1-p)) \nwspace \sqrt{(1\,{-}\,\alpha) n}}\right) 
						\cdot \left(\frac{p^{1-\alpha}}
						{ (1-\alpha)^{1-\alpha} \ \alpha^\alpha} \right)^n\!.
			\end{align*}
	\end{enumerate}
\end{theorem}

The bounds in the two cases are very different in nature.
They are visualized in Figure~\ref{fig:plots}
showing the expected size of the minimization 
as a function of the number of edges $m$ and of $\alpha$.
To distinguish the cases also in writing, we
use the term \emph{linear regime} if $m$ is between
$1$ and $1/(1-p)^{(1-p)n}$, corresponding to $0 \le \alpha \le 1-p$,
likewise, we refer to $m$ being between
$1/(1-p)^{(1-p) n}$ and $1/(1-p)^n$, i.e., $1-p \le \alpha \le 1$,
as the \emph{information-theoretic regime}.
In Statement~($\ref{clause:main_thm_information-theoretic_range}$) of the above theorem,
the leading constant may depend on the choice of $\varepsilon$.

In order to derive Theorem~\ref{thm:properties_Bnpk} from Theorem~\ref{thm:binomial_characterization},
we tighten the Chernoff--Hoeffding inequality on the binomial distribution function.
Let
$\Dkl(x \,{\|}\, y)
	=  -x \nwspace \log_2\!\left( \frac{y}{x} \right) 
		- (1-x) \nwspace \log_2\!\left(\frac{1-y}{1-x} \right)$
denote the binary Kullback--Leibler divergence.

\wormhole{thm:better_Hoeffding}
\begin{theorem}
\label{thm:better_Hoeffding}
	Let $n$ be a positive integer, $0 < p < 1$ a non-trivial probability, and
	$X  \sim \Bin(n,p)$ a binomial variable.
	Suppose the function $x = x(n)$ takes real values in the interval
	$[\varepsilon, \, 1\,{-}\,\varepsilon]$ for some $\varepsilon > 0$.
	Let $\varphi$ and $\psi$ denote the functions
	\begin{equation*}
		\varphi(n,p,x) = \min\!\left( 1, \,\frac{1}{(p-x) \nwspace \sqrt{xn}} \right)
		\ \, \text{and}\quad
		\psi(n,p,x) = \min\!\left( 1, \,\frac{1}{(x-p) \nwspace \sqrt{(1-x)n}} \right)\!,
	\end{equation*}
	
	\noindent
	with additionally $\varphi(n,p,p) = 1$ and $\psi(n,p,p) = 1$.\vspace*{.75em}

	\noindent
	There exist constants $C_1, C_2, C_3, C_4 > 0$,
	independent of $n$ and $x$ but possibly dependent on $p$ and $\varepsilon$,
	such that the following statements hold for all $n$ sufficiently large.

	\begin{enumerate}[\textup{(}$i$\textup{)}]
		\item If $x \le p$, then
			$C_1 \nwspace \varphi \cdot 2^{-\!\Dkl(x \,{\|}\, p) \nwspace n} \le \Pb[X \le xn]
				\le C_2 \nwspace \varphi \cdot 2^{-\!\Dkl(x \,{\|}\, p) \nwspace n}$.
		\item If $x \ge p$, then
			$C_3 \nwspace \psi \cdot 2^{-\!\Dkl(x \,{\|}\, p) \nwspace n} \le \Pb[X \ge xn]
				\le C_4 \nwspace \psi \cdot 2^{-\!\Dkl(x \,{\|}\, p) \nwspace n}$.
	\end{enumerate}
\end{theorem}

\section{Preliminaries and Notation}
\label{sec:prelims}

\subparagraph*{Multi-Hypergraphs.}

A \emph{hypergraph} on the vertex set $[n] = \{1, \dots, n \}$ is a set of subsets
$\Hyp \,{\subseteq}\, \mathcal{P}([n])$, called the \mbox{\emph{(hyper-)edges.}}
If $\Hyp$ is a multiset instead, we have a \emph{multi-hypergraph}.
We do not allow multiple copies of the same vertex in one edge.
The \emph{minimization} of a hypergraph $\Hyp$ 
is the collection of its inclusion-wise minimal edges,
$\min(\Hyp) = \{ E \in \Hyp \mid \forall E' \in \Hyp \colon E' \subseteq E \Rightarrow E' = E\}$.
We extend this notion to multi-hypergraphs by requiring
that, whenever a minimal edge has multiple copies,
only one of them is included in the minimization.
This way $\min(\Hyp)$ is always a mere hypergraph (a set).
For a multi-hypergraph $\Hyp$, we use $|\Hyp|$ to denote
the total number of edges counting multiplicities,
and $\|\Hyp\|$ for the number of distinct edges,
that is, the cardinality of the support of $\Hyp$.
Evidently, we have $|\min(\Hyp)| \le \|\Hyp\| \le |\Hyp|$.

\vspace*{-.75em}
\subparagraph*{Information Theory.}

We intend the expressions $0 \nwspace\log_a 0$ and $0 \nwspace\log_a (\frac{0}{0})$
to both mean $0$ for any positive real base $a$.
Note that this convention implies $0^0 = a^{0 \log_a 0} = 1$
and $(\frac{0}{0})^0 = 1$.
We use $\ld x$ for the binary (base-$2$) logarithm of $x$.
The (binary) \emph{entropy function} $\Hb$ is defined for all probabilities $x$ as
\begin{equation*}
	\Hb(x) = - x \nwspace \ld x - (1-x) \nwspace \ld(1-x).
\end{equation*}

\noindent
The entropy function is the Shannon entropy (equivalently, the R{\'e}nyi entropy of order $1$)
of the Bernoulli distribution with parameter $x$.
In the notation of the previous sections, we have $\Hb(x) = H((x, 1-x))$.
The entropy function is symmetric around $\sfrac{1}{2}$ with $\Hb(x) = \Hb(1-x)$.
On the open unit interval, $\Hb$ is positive and 
strictly concave, it has its maximum at $\sfrac{1}{2}$ with value $\Hb(\sfrac{1}{2}) = 1$.
The entropy power
\begin{equation*}
	2^{\Hb(x)} = \frac{1}{x^x \nwspace (1-x)^{1-x}}
\end{equation*}

\noindent
is called the \emph{perplexity} of $x$.
We use it to estimate binomial coefficients, see~\cite{Cover2006InformationTheory}.

\begin{lemma}
\label{lem:lowerbound_binCoeff}
	Let $n$ be a positive integer and $0 < x <1 $ a rational such that $xn$ is an integer, then
	\begin{equation*}
		\frac{2^{\Hb(x) n}}{\sqrt{8 \nwspace n \nwspace x (1-x)}}
			\le \binom{n}{xn}
			\le \frac{2^{\Hb(x) n}}{\sqrt{\pi \nwspace n \nwspace x (1-x)}}.
	\end{equation*}
\end{lemma}

For any two probabilities $x$ and $y$,
the (binary) \emph{Kullback--Leibler divergence}\footnote{%
	The divergence is sometimes called \emph{relative entropy},
	we avoid this term due to ambiguities, see~\cite{Cover2006InformationTheory}.}
between Bernoulli distributions with respective parameters $x$ and $y$ is
\begin{equation*}
	\Dkl(x \,{\|}\, y)
		=  -x \nwspace \ld\!\left( \frac{y}{x} \right) - (1-x) \nwspace \ld\!\left(\frac{1-y}{1-x} \right)\!.
\end{equation*}

\noindent
The function $\Dkl(x \,{\|}\, y)$ is convex in both $x$ and $y$,
and attains its minimum $0$ for $x = y$.
The divergence observes $\Dkl(x \,{\|}\, y) = \Dkl(1 - x \,{\|}\, 1-y)$.
Its partial derivative with respect to $x$ is
\begin{equation*}
	\frac{\partial}{\partial x} \Dkl(x \,{\|}\, y) 
		= \ld \!\left( \frac{x}{1-x}\, \frac{1-y}{y} \right)\!.
\end{equation*}

\noindent
Next, we establish the fact that the divergence scales quadratically
in the difference $y - x$. 

\begin{lemma}
\label{lema:div_quadratic}
	Let $0 < x \le y < 1$ be two non-trivial probabilities.
	Define $t^+$ to be the maximizer of $t(1-t)$ over the interval $[x,y]$,
	and $t^-$ the minimizer.
	Then, it holds that
	\begin{equation*}
		\frac{(y-x)^2}{t^+(1-t^+)}
			\le 2\ln(2) \cdot \Dkl(x \,{\|}\, y) \le
			\frac{(y-x)^2}{t^-(1-t^-)}.
	\end{equation*}
\end{lemma}

\begin{proof}
	Let $\varepsilon = y - x$.
	$\Dkl(y - \varepsilon \,{\|}\, y)$ is two-times differentiable with respect to $\varepsilon$
	with derivatives
	\begin{equation*}
		\frac{\partial}{\partial \varepsilon} \Dkl(y - \varepsilon \,{\|}\, y)
			= \ld\!\left(\frac{y}{1-y} \nwspace \frac{1-y+\varepsilon}{y-\varepsilon}\right)
		\ \text{and} \ \ \
		\frac{\partial^2}{\partial \varepsilon^2} \Dkl(y - \varepsilon \,{\|}\, y)
			= \frac{1}{\ln 2} \, \frac{1}{(y-\varepsilon)(1-y+\varepsilon)}.
	\end{equation*}
	
	\noindent
	Both the divergence itself as well as its first derivative vanish at $\varepsilon = 0$.
	By Taylor's theorem (with the Lagrange form of the remainder),
	there exists an $\xi$ with $0 \le \xi \le \varepsilon$ such that
	\begin{equation*}
		\Dkl(y - \varepsilon \,{\|}\, y)
			= \frac{\varepsilon^2}{2!} \cdot \left. \frac{\partial^2}{\partial \varepsilon^2}
				\Dkl(y - \varepsilon \,{\|}\, y) \nwspace \right|_{\varepsilon = \xi}
			= \frac{\varepsilon^2}{2 \ln 2 \, (y-\xi)(1-y+\xi)}.
	\end{equation*}	
	
	\noindent
	The lemma follows from $y-\xi$ ranging over $[x,y]$.
\end{proof}

We often use the following quantity derived from the divergence, resembling the perplexity.
\begin{equation*}
	2^{-\!\Dkl(x \,{\|}\, y)} = 2^{\Hb(x)} \cdot y^x (1-y)^{1-x} = \left(\frac{y}{x}\right)^x \left(\frac{1-y}{1-x}\right)^{1-x}\!\!.
\end{equation*}

\noindent
The next lemma is useful when relating quantities of this kind for different parameters.

\begin{lemma}
\label{lem:monotone_div}
	Let $0 \le x \le y \le z \le 1$ be three probabilities, then
	\begin{equation*}
		2^{-\!\Dkl(x \,{\|}\, z)} = 
			\left(\frac{y}{1-y} \, \frac{1-z}{z}\right)^{y-x}\!
				2^{-\!\Dkl(x \,{\|}\, y)} \cdot 2^{-\!\Dkl(y \,{\|}\, z)}.
	\end{equation*}
	
	\noindent
	In particular, for any fixed $z$,
	$2^{-\!\Dkl(x \,{\|}\, z)}$ is non-decreasing in $x$ as long as $x \le z$.
\end{lemma}

\begin{proof}
	The convention $(\frac{0}{0})^0 = 1$
	ensures that $(\frac{y}{1-y} \frac{1-z}{z})^{y-x}$ is well-defined even for $z=0$.
	\begin{align*}
		\frac{2^{-\!\Dkl(x \,{\|}\, z)}}{2^{-\!\Dkl(y \,{\|}\, z)}}
			&= \frac{\left(\frac{z}{x}\right)^x \left(\frac{1-z}{1-x}\right)^{1-x}}
				{\left(\frac{z}{y}\right)^y \left(\frac{1-z}{1-y}\right)^{1-y}}
			 = \frac{\left(\frac{z}{x}\right)^x \left(\frac{1-z}{1-x}\right)^{1-x}}
				{\left(\frac{z}{y}\right)^{y-x} \left(\frac{z}{y}\right)^x
					\left(\frac{1-z}{1-y}\right)^{1-x} \left(\frac{1-z}{1-y}\right)^{x-y}}\\[.5em]
			&= \frac{1}{\left(\frac{z}{y}\right)^{y-x} \left(\frac{1-z}{1-y}\right)^{x-y}} 
				\left(\frac{y}{x}\right)^x \left(\frac{1-y}{1-x}\right)^{1-x}
			 = \left(\frac{y}{1-y} \, \frac{1-z}{z}\right)^{y-x}
			 	2^{-\!\Dkl(x \,{\|}\, y)}.
	\end{align*}
	
	\noindent
	The monotonicity follows from the last two factors being at most $1$
	due to $x \le y \le z$.
\end{proof}

\vspace*{-1em}
\subparagraph*{Polynomials of Probabilities.}

We regularly estimate expressions of the form $(1-x)^n$
where $x$ is a probability.
The first inequality is taken from~\cite{Motwani1995RandAlg}.

\begin{lemma}
\label{lem:Motwani_and_Raghavan_bound}
	Let $n$ be a positive integer and $x$ a real number such that $|x| \le n$, then
	\begin{equation*}
		e^{x} \left(1 - \frac{x^2}{n}\right) \le \left( 1 + \frac{x}{n}\right)^n\!.
	\end{equation*}
\end{lemma}

\noindent
We reach rather tight bounds on $(1-x)^n$ by substituting $x$ for $-nx$ above,
and combining this with the simple fact that $(1+x) \le e^{x}$ holds for all $x$.

\begin{corollary}
\label{cor:tight_bound_polynom}
	Let $n$ be a non-negative integer and $x$ a probability, then
	\begin{equation*}
		e^{-nx} \left( 1 - n \nwspace x^2\right) \le (1 - x)^n \le e^{-nx}\!.
	\end{equation*}
\end{corollary}

\noindent
The next inequalitiy was given by Badkobeh, Lehre, and Sudholt~\cite{Badkobeh:2015}.

\begin{lemma}[Lemma~10 in~\cite{Badkobeh:2015}]
\label{lem:Badkobeh_bound}
	Let $n$ be a non-negative integer and $x$ a probability, then
	\begin{equation*}
		\frac{nx}{1+nx} \le 1-(1-x)^n \le nx.
	\end{equation*}
\end{lemma}

We prepare the following lemma on conditional probabilities of series of events for later.

\begin{lemma}
\label{lem:znZN}
	Consider a random experiment with three outcomes $A$, $B$, and $C$ with $\Pb[ B ] > 0$.
	In a series of $m$ i.i.d.\ trials, let $A_j$ denote the event
	that the outcome of the $j$-th trial is $A$, same with $B$.	
	Then, it holds that
	$\Pb[\nwspace \forall j \le m \colon \neg A_j \mid \exists k \le m \colon B_k \nwspace]
	\le \Pb[\nwspace \forall j \le m \colon \neg A_j \mid B_m \nwspace]$.
\end{lemma}

\begin{proof}
	The case $\Pb[B] = 1$ is trivial, thus assume $0 < \Pb[B] < 1$.
	The assertion in the lemma is equivalent to
	\begin{equation}
	\label{eq:equivalent_probabilities}
		\Pb[\nwspace \forall j \le m \colon \neg A_j
			\mid \neg B_m \wedge (\exists k < m \colon B_k) \nwspace]
		\le \Pb[\nwspace \forall j \le m  \colon \neg A_j \mid B_m \nwspace].
	\end{equation}

	\noindent	
	To see this, observe that for any four reals $x$, $y$, $z$, $w$
	such that $y$, $w$, and $y+w$ are all non-zero,
	$\frac{x+z}{y+w} \le \frac{z}{w}$ holds if and only if $\frac{x}{y} \le \frac{z}{w}$ does.
	The event $[\nwspace \exists k \le m \colon B_k \nwspace]$
	can be partitioned into
	$[ \nwspace \neg B_m \wedge (\exists k < m \colon B_k) \nwspace]$
	and $[ B_m ]$, giving
	\begin{multline*}
		\Pb[ \nwspace \forall j \le m \colon \neg A_j
			\mid \exists k \le m  \colon B_k \nwspace]
			= \frac{\Pb[ \nwspace (\forall j \le m  \colon \neg A_j) 
				\wedge (\exists k \le m  \colon B_k) \nwspace ]}
				{\Pb[\nwspace \exists k \le m  \colon B_k \nwspace]}\\
			= \frac{\Pb[ \nwspace (\forall j \le m  \colon \neg A_j)
					\wedge \neg B_m \wedge (\exists k < m \colon B_k) \nwspace ]
				+ \Pb[ \nwspace (\forall j \le m  \colon \neg A_j) \wedge B_m \nwspace ]}
				{\Pb[ \nwspace \neg B_m \wedge (\exists k < m \colon B_k) \nwspace ] 
					+ \Pb[\nwspace B_m \nwspace]}\nwspace.
	\end{multline*}
	
	\noindent
	Applying the observation to the real numbers
	$x = \Pb[ \nwspace (\forall j \le m \colon \neg A_j) \wedge \neg B_m 
		\wedge (\exists k < m \colon B_k) \nwspace ]$,
	$y = \Pb[ \nwspace \neg B_m \wedge (\exists k < m \colon B_k) \nwspace ]$,
	$z = \Pb[ \nwspace (\forall j \le m  \colon \neg A_j) \wedge B_m \nwspace ]$,
	and $w = \Pb[ B_m ]$ gives the equivalence.
	
	The actual lemma is proven by induction over $m$.
	The case $m = 1$ is trivial as both sides simplify to $\Pb[\neg A_1 \mid B_1]$.	
	Suppose that $\Pb[ \nwspace \forall j < m \colon \neg A_j 
	\mid \exists k < m \colon B_k \nwspace ]
	\le \Pb[ \nwspace \forall j < m \colon \neg A_j \mid B_{m-1} \nwspace ]$
	holds, it is sufficient to conclude (\ref{eq:equivalent_probabilities}).
	The independence of the trials imply
	\begin{align*}
		\hspace*{-2.66em}													% avoid an overfull box
		\Pb[ \nwspace \forall j \le m \colon \neg A_j
			\mid \neg B_m \wedge (\exists k < m \colon B_k) \nwspace ]
			&= \frac{\Pb[ \nwspace (\forall j \le m \colon \neg A_j)
				\wedge \neg B_m \wedge (\exists k < m \colon B_k) \nwspace ]}
				{\Pb[\nwspace \neg B_m \wedge (\exists k < m \colon B_k) \nwspace ]}\\
			&= \frac{\Pb[ \nwspace \neg A_m \wedge \neg B_m \nwspace ]
				\cdot \Pb[\nwspace (\forall j < m \colon \neg A_j)
					\wedge (\exists k < m \colon B_k) \nwspace ]}
				{\Pb[\nwspace \neg B_m \nwspace] 
					\cdot \Pb[ \nwspace \exists k < m \colon B_k \nwspace ]}\\
			&= \Pb[ \nwspace \neg A_m \mid \neg B_m \nwspace ] 
				\cdot \Pb[ \nwspace \forall j < m \colon \neg A_j 
					\mid \exists k < m \colon B_k \nwspace].
	\end{align*}
	
	\noindent
	By induction, this is at most $\Pb[\nwspace \neg A_m \mid \neg B_m \nwspace ] \cdot
	\nwspace \Pb[\nwspace \forall j < m \colon \neg A_j \mid B_{m-1} \nwspace]$.
	The probabilities of the outcomes do not change over the trials,
	and also event $B_m$ implies $\neg A_m$.
	Therefore,
	\begin{multline*}
		\Pb[ \nwspace \neg A_m \mid \neg B_m \nwspace ] \cdot
			\nwspace \Pb[ \nwspace \forall j < m \colon \neg A_j | B_{m-1} \nwspace]
			= \frac{1-\Pb[A_m] - \Pb[B_m]}{1-\Pb[B_m]} 
				\ \frac{(1- \Pb[A_m])^{m-2} \cdot \Pb[B_{m-1}]}{\Pb[B_{m-1}]}\\
			= \left(1 - \frac{\Pb[A_m]}{1-\Pb[B_m]}\right) \cdot \Big(1- \Pb[A_m]\Big)^{m-2}
			\le \Big(1- \Pb[A_m]\Big)^{m-1}
			= \Pb[\nwspace \forall j \le m \colon \neg A_j \mid B_m \nwspace ]. \qedhere
	\end{multline*}
\end{proof}

\section{The Chernoff--Hoeffding Theorem}
\label{sec:Chernoff--Hoeffding}

Most of the concentration bounds on the binomial distribution
that are used in many different fields of science
can be traced back to the Chernoff--Hoeffding theorem~\cite{Hoeffding63Inequalities,Dubhashi09ConcentrationOfMeasure}.
It employs the Kullback--Leibler divergence to bound the probability
that a random variable $X \sim \Bin(n,p)$
deviates from its expectation $\Ev[X] = pn$.
The theorem states for all $x$ with $0 \le x \le p$ that
\begin{equation*}
	\Pb[X \le xn] \le 2^{-\!\Dkl(x \,{\|}\, p) \nwspace n}
		= \Big(\frac{p}{x}\Big)^{xn} \left(\frac{1-p}{1-x}\right)^{(1-x)n}\!\!.
\end{equation*}

\noindent
It follows easily that also
$\Pb[ X \,{\ge}\, x n]  \le 2^{-\!\Dkl(x \,{\|}\, p) \nwspace n}$
is true for all $p \le x \le 1$.
Several weaker but more practical inequalities have been derived from the Chernoff--Hoeffding theorem,
colloquially summarized as Chernoff bounds~\cite{Mitzenmacher17Probablility,Doerr20Evolutionary}.
Cramér's theorem asserts that
the exponent $\Dkl(x \,{\|}\, p)$ is asymptotically tight~\cite{Cramer38SurUn,Deuschel84LargeDeviations}.
Any improvement of this inequality can therefore be at most subexponential in $n$.
Stirling's approximation, see~\cite{Ash90InformationTheory},
gives the following lower bound assuming that the product $xn$ is an integer,
\begin{equation*}
	\Pb[X \le xn] 
		\ge \frac{1}{\sqrt{8 \nwspace n \nwspace x(1-x)}} \cdot 2^{-\!\Dkl(x \,{\|}\, p) \nwspace n}.
\end{equation*}

There is an obvious gap of order $\sqrt{n}$ between the two estimates.
This section is dedicated to closing this gap by proving Theorem~\ref{thm:better_Hoeffding}.
It sharpens both inequalities, making them tight up to constant factors.
While this improvement will be a valuable tool in our study of hypergraphs,
it has applications also in many other areas.
The proof is split into two major parts.
The first one treats those $x$ for which the product $xn$ is an integer,
the second one extends this to the general case.
The parts are further subdivided depending on the limiting behavior of $x = x(n)$,
namely, whether it converges to $p$ or not.

\subsection{Integral Case}
\label{subsec:integral_case}

For the integral case,
we first use a result by Klar~\cite{Klar00TailProbabilities}
about the connection between the distribution function
and the probability mass function (PMF) of the binomial distribution.

\begin{lemma}[Proposition {1}$\nwspace$(c) in \cite{Klar00TailProbabilities}]
\label{lem:Klar}
	Let $n$ be a positive integers, $p \neq 0$ a probability,
	and $X \sim \Bin(n,p)$ a binomial variable.
	For all non-negative integers $k \le pn$, it holds that
	\begin{equation*}
		1 \le \frac{\Pb[X \le k]}{\Pb[X = k]}
		\le \frac{p (n+1-k)}{n+1-k - (n+1)(1-p)}.
	\end{equation*}
\end{lemma}

\noindent
We combine this with the perplexity bound on the binomial coefficient
in Lemma~\ref{lem:lowerbound_binCoeff}.
The respective lower bounds in the next lemma
where known before, see for example the textbook by Ash~\cite[Lemma~4.7.2]{Ash90InformationTheory},
we reprove them here en passant.
The upper bounds are novel and we will later use them in our proof of 
the integral case of Theorem~\ref{thm:better_Hoeffding}.

\begin{lemma}
\label{lem:integral_bounded_away}
	Let $n$ be a positive integer, $0 < p < 1$ a non-trivial probability, 
	and $X \sim \Bin(n,p)$ a binomial variable.
	Suppose $0 < x < 1$ is a rational such that $xn$ is an integer.
	\begin{enumerate}[\textup{(}$i$\textup{)}]
		\item\label{clause:lem_bounded_away_below_mean}
			If $x < p$, then
			\begin{equation*}
			\frac{1}{\sqrt{8 \nwspace n \nwspace x (1-x)}}
				\cdot 2^{-\!\Dkl(x \,{\|}\, p) \nwspace n}
				\le \Pb[ X \le x n]
				\le \frac{ p \nwspace \sqrt{1-x}}{(p - x) \nwspace \sqrt{\pi \nwspace xn}}
					\cdot  2^{-\!\Dkl(x \,{\|}\, p) \nwspace n}.
			\end{equation*}
		\item\label{clause:lem_bounded_away_above_mean}
			If $p < x$, then
			\begin{equation*}
			\frac{1}{\sqrt{8 \nwspace n \nwspace x (1-x)}}
				\cdot 2^{-\!\Dkl(x \,{\|}\, p) \nwspace n}
				\le \Pb[ X \ge x n]
					\le \frac{(1-p) \nwspace \sqrt{x}}{(x - p) \nwspace \sqrt{\pi \nwspace (1-x)n}}
							\cdot 2^{-\!\Dkl(x \,{\|}\, p) \nwspace n}.
			\end{equation*}
	\end{enumerate}
\end{lemma}

\begin{proof}
	Applying the first statement to the complementary variable
	$\overline X \sim \Bin(n,1\,{-}\,p)$
	implies the second statement since
	$\Pb[ X \ge x n ] = \Pb[\nwspace \overline X \le (1\,{-}\,x)n]$.
	Hereby, we use that the Kullback--Leibler divergence observes
	$\Dkl(1\,{-}\,x \,{\|}\, 1\,{-}\,p) = \Dkl(x \,{\|}\, p)$.
	We are left to prove the first statement.
	
	Lemma~\ref{lem:lowerbound_binCoeff} gives the following error bounds on the PMF
	$\Pb[X = xn] = \binom{n}{xn} \cdot p^{xn}(1-p)^{(1-x) n}$.
	\begin{equation*}
		\frac{1}{\sqrt{8 \nwspace n \nwspace x(1-x)}}
			\le \frac{\Pb[X = xn]} {2^{\Hb(x) n} \cdot p^{xn} (1-p)^{(1-x)n}}
			= \frac{\Pb[X = xn]} {2^{- \Dkl(x \,{\|}\, p) \nwspace n}}
			\le \frac{1}{\sqrt{\pi \nwspace n \nwspace x(1-x)}}.
	\end{equation*}	
	The lower bounds follows immediately from $\Pb[X \le xn] \ge \Pb[X = xn]$.
	
	For the upper bound, we use Lemma~\ref{lem:Klar} at the integer position $k = xn$.
	Let $f_{n,xn}(p)$ denote the resulting bound on the ratio $\Pb[X \le xn]/\Pb[X =xn]$,
	that is,
	\begin{equation*}
		f_{n,xn}(p) = \frac{p \nwspace (n+1-xn)}{n+1-xn - (n+1)(1-p)}
			= \frac{p \nwspace (n+1 - xn)}{p \nwspace (n+1) - xn}
			= \frac{p \nwspace (1-\frac{xn}{n+1})}{p-\frac{xn}{n+1}}.
	\end{equation*}
	
	\noindent
	We claim that for all $x$ and $p$ with $x < p$, the function $f_{n,xn}(p)$ is increasing in $n$.
	We show this by verifying that the (partial) discrete derivative $\Delta_n( f_{n,xn})$
	with respect to $n$ is positive.
	\begin{align*}
		\Delta_n( f_{n,xn})(p)
			&= f_{n+1,x(n+1)}(p) - f_{n,xn}(p)
			 = \frac{p \nwspace (n+2-x(n+1))}{p \nwspace (n+2) - x(n+1)}
				- \frac{p \nwspace (n+1-xn)}{p \nwspace (n+1) - xn}\\
			&= \frac{p \nwspace (1-p) \nwspace x}{((p-x) \nwspace n+p) \cdot ((p-x) \nwspace n+2p-x)}
			 > 0.
	\end{align*}
	
	The function $f_{n,xn}(p)$ thus converges from below
	to $p \nwspace (1-x)/(p-x)$ as $n$ increases,
	giving an upper bound on $\Pb[X \le xn]/\Pb[X =xn]$ for all $n$.
	Multiplying with the error bounds and the divergence completes the proof.
\end{proof}

The upper bounds above are already
very close to the desired ones of Theorem~\ref{thm:better_Hoeffding}.
In fact, we will see that the lemma is enough to conclude
$\Pb[X \le xn] \le \varphi \cdot 2^{- \Dkl(x \,{\|}\, p) \nwspace n}$
if $xn$ is an integer and $\varphi = \min(1, 1/((p-x) \sqrt{xn}))$.
The lower bound in Statement~($\ref{clause:lem_bounded_away_below_mean}$), however,
matches the upper one only
if the $x = x(n)$ is bounded away from $p$ for all $n$.
More work is needed for the case $x \to p$.
It has already been useful to have a good estimate for the ratio $\Pb[X \le k]/\Pb[X = k]$.
Unfortunately, Lemma~\ref{lem:Klar} gives only a trivial lower bound.
We strengthen this in the next lemma,
Lemma~\ref{lem:integral_converge_to_p} then shows how this
translates into a stronger lower bound on the binomial distribution function.
Finally, Lemma~\ref{lem:integral_case}  combines all results of this section
into a version of the Chernoff--Hoeffding theorem,
which is tight whenever the product $xn$ is an integer.

\begin{lemma}
\label{lem:ratio_lower_bound}
	Let $n$ be a positive integer, $0 < p < 1$ a non-trivial probability,
	and $X \sim \Bin(n,p)$ a binomial variable.
	Then, for all non-negative integers $i$ and $k$ with $i \le k \le pn$, it holds that
	\begin{equation*}
		\frac{\Pb[X \le k]}{\Pb[X = k]}
		\ge (k-i+1) \left( 1- \frac{pn - i}{pn(1- \frac{k}{n})}\right)^{k-i}.
	\end{equation*}
\end{lemma}
\vspace*{1em}

\begin{proof}
	The PMF of $X$ is increasing for arguments smaller than $pn$,
	therefore
	\begin{equation*}
		\frac{\Pb[X \le k]}{\Pb[X = k]} 
			= \sum_{j=0}^k \frac{\Pb[X = j]}{\Pb[X = k]}
			\ge \sum_{j=i}^k \frac{\Pb[X = j]}{\Pb[X = k]} 
			\ge (k-i+1) \,\frac{\Pb[X = i]}{\Pb[X = k]}.
	\end{equation*}
	
	\noindent
	The last ratio is lower-bounded by
	\begin{equation*}
		\frac{\Pb[X = i]}{\Pb[X = k]} 
			= \frac{k! \, (n-k)!}{i! \, (n-i)!} \left(\frac{1-p}{p}\right)^{k-i}
			= \prod_{\ell=1}^{k-i} \frac{i+\ell}{n-i-\ell+1} \cdot \left(\frac{1-p}{p}\right)^{k-i}
			\ge \left( \frac{i}{n-i} \frac{1-p}{p} \right)^{k-i}\!\!.
	\end{equation*}
	
	\noindent
	For the base of the last expression we get
	\begin{equation*}
		\frac{i}{n-i} \frac{1-p}{p} 
			= \frac{pn - pi - pn + i}{pn-pi} 
			= 1 - \frac{pn - i}{pn-pi}
			= 1- \frac{pn - i}{n p(1-\frac{i}{n})}
			\ge 1- \frac{pn - i}{n p(1-\frac{k}{n})}. \qedhere
	\end{equation*}
\end{proof}

The first factor $k-i+1$ increases as $i$ gets smaller
while at the same time the second factor decreases.
Therefore, in order to apply Lemma~\ref{lem:ratio_lower_bound},
one has to choose a balancing cut-off point.
We do so in the proof of the following lemma.

\begin{lemma}
\label{lem:integral_converge_to_p}
	Let $n$ be a positive integer, $0 < p< 1$ a non-trivial probability,
	and $X \sim \Bin(n,p)$ a binomial variable.
	Suppose $0 < x < 1$ is a rational such that $xn$ is an integer.
	\begin{enumerate}[\textup{(}$i$\textup{)}]
		\item\label{clause:lem_converge_below_mean}
			If $x < p$, then
			\begin{equation*}
				\Pb[ X \le x n]
				\ge \frac{p\nwspace \sqrt{1-x}}{16 \nwspace \sqrt{2}} 
				\cdot \min\!\left( 1, \,\frac{1}{(p-x) \nwspace \sqrt{xn}} \right)
				\cdot 2^{-\!\Dkl(x \,{\|}\, p) \nwspace n}.
			\end{equation*}
		\item\label{clause:lem_converge_above_mean}
			If $p < x$, then
			\begin{equation*}
				\Pb[ X \ge x n]
				\ge \frac{(1-p) \nwspace \sqrt{x}}{16 \nwspace \sqrt{2}} 
				\cdot \min\!\left( 1, \,\frac{1}{(x-p) \nwspace \sqrt{(1-x)n}} \right)
				\cdot 2^{-\!\Dkl(x \,{\|}\, p) \nwspace n}.
			\end{equation*}
	\end{enumerate}
\end{lemma}

\begin{proof}
	The second statement follows from the first as in Lemma~\ref{lem:integral_bounded_away}.
	Define an auxiliary integer function $g = g(n,p,x)$ as
	\begin{equation*}
		g(n,p,x) = \left\lfloor \frac{p(1-x)}{2} 
			\cdot \min\!\left( \sqrt{xn}, \,\frac{1}{p-x} \right)\nwspace\right\rfloor\!.
	\end{equation*}
	
	\noindent
	Applying Lemma~\ref{lem:ratio_lower_bound} at position $k = xn$ 
	with the cut-off point $i = xn - g$ gives
	\begin{equation}
	\label{eq:distribution_lower_bound}
		\Pb[ X \le x n] \ge (g+1) \left( 1- \frac{pn - xn + g}{pn(1- x)}\right)^{g} \cdot \Pb[X = xn].
	\end{equation}
	
	We want to lower-bound the middle factor in (\ref{eq:distribution_lower_bound}) by a constant.
	Bernoulli's inequality gives
	\begin{equation*}
		\left( 1- \frac{pn - xn + g}{pn(1- x)}\right)^{g}
			= \left( 1- \frac{p-x + \frac{g}{n}}{p(1- x)}\right)^{g}
			\ge 1- \frac{g(p-x) + \frac{g^2}{n}}{p(1- x)}.
	\end{equation*}
	
	\noindent
	We claim that the numerator $g(p-x) + g^2/n$ is at most $3 \nwspace p(1-x)/4$.
	We split the argument depending on the relative size of $\sqrt{xn}$ and $1/(p-x)$.
	If $\sqrt{xn} \ge 1/(p-x)$, then we have 
	$g =  \lfloor \, p(1-x)/ 2(p-x) \rfloor$ and thus
	\begin{equation*}
		g(p-x) + \frac{g^2}{n}
			\le \frac{p(1-x)}{2} + \frac{p^2 (1-x)^2}{4}  \cdot \frac{1}{(p-x)^2 \nwspace n}
			\le \frac{p(1-x)}{2} + \frac{p^2 (1-x)^2}{4} \cdot x.
	\end{equation*}
	
	\noindent
	Conversely, if $\sqrt{xn} \le 1/(p-x)$, then
	$g =  \lfloor \, p(1-x) \nwspace \sqrt{xn}/ 2 \rfloor$ and
	\begin{equation*}
		g(p-x) + \frac{g^2}{n}
			\le \frac{p(1-x)}{2} \cdot (p-x) \sqrt{xn}+ \frac{p^2 (1-x)^2}{4}  \cdot \frac{xn}{n}
			\le \frac{p(1-x)}{2} + \frac{p^2 (1-x)^2}{4} \cdot x.
	\end{equation*}
%	Since $x$, $1-x$, and $p$ are all smaller than $1$, 
	The last expressions of both inequalities are the same
	and can be bounded by $3 \nwspace p(1-x)/4$.
	The middle factor is therefore at least a constant since
	\begin{equation*}
		1- \frac{g(p-x) + \frac{g^2}{n}}{p(1- x)}
			\ge 1 - \frac{3 \nwspace p(1-x)}{4} \ \frac{1}{p(1-x)} = \frac{1}{4}.
	\end{equation*}
	
	\noindent
	Reinserting this into Inequality~(\ref{eq:distribution_lower_bound})
	and applying the definition of $g$ and Lemma~\ref{lem:lowerbound_binCoeff} gives the result.
	\begin{align*}
		\Pb[ X \le x n] 
			&\ge \frac{g+1}{4} \cdot \Pb[X = xn]
			 \ge \frac{p(1-x)}{8} 
				\cdot \min\!\left( \sqrt{xn}, \,\frac{1}{p-x} \right) \cdot \Pb[X = xn]\\
			&\ge \frac{p(1-x)}{8} 
				\cdot \min\!\left( \sqrt{xn}, \,\frac{1}{p-x} \right)
				\cdot \frac{1}
					{\sqrt{8 \nwspace n \nwspace x (1-x)}} \cdot 2^{-\!\Dkl(x \,{\|}\, p) \nwspace n}\\
			&= \frac{p\nwspace \sqrt{1-x}}{16 \nwspace \sqrt{2}} 
				\cdot \min\!\left( 1, \,\frac{1}{(p-x) \nwspace \sqrt{xn}} \right)
				\cdot 2^{-\!\Dkl(x \,{\|}\, p) \nwspace n}. \qedhere
	\end{align*}
\end{proof}

Next, we prove Theorem~\ref{thm:better_Hoeffding} for the case that $xn$ is an integer
by combining the results above.
We emphasize the facts that Lemma~\ref{lem:integral_case} holds for all positive integers $n$,
not only asymptotically,
and $x$ may range over the whole interval $[0,1]$.

\begin{lemma}[integral case of Theorem~\ref{thm:better_Hoeffding}]
\label{lem:integral_case}
	Let $n$ be a positive integer, $0 < p < 1$ a non-trivial probability,
	and	$X  \sim \Bin(n,p)$ a binomial variable.
	Suppose $x = x(n)$ takes rational values in the unit interval such that $xn$ is an integer.
	Let $\varphi$ and $\psi$ denote the functions
	\begin{equation*}
		\varphi(n,p,x) = \min\!\left( 1, \,\frac{1}{(p-x) \nwspace \sqrt{xn}} \right)
		\quad\text{and}\quad
		\psi(n,p,x) = \min\!\left( 1, \,\frac{1}{(x-p) \nwspace \sqrt{(1-x)n}} \right)\!,
	\end{equation*}

	\noindent
	with additionally $\varphi(n,p,0) = \varphi(n,p,p) = 1$
	and $\psi(n,p,1) = \psi(n,p,p) = 1$.
	\begin{enumerate}[\textup{(}$i$\textup{)}]
		\item\label{clause:integral_case_below_mean} If $x \le p$, then
			$\frac{p \nwspace \sqrt{1-p}}{16 \sqrt{2}}
				\cdot \varphi \cdot 2^{-\!\Dkl(x \,{\|}\, p) \nwspace n} \le \Pb[X \le xn]
				\le \varphi \cdot 2^{-\!\Dkl(x \,{\|}\, p) \nwspace n}$.
		\item\label{clause:integral_case_above_mean} If $x \ge p$, then
			$\frac{(1-p) \sqrt{p}}{16 \sqrt{2}}
				\cdot \psi \cdot 2^{-\!\Dkl(x \,{\|}\, p) \nwspace n} \le \Pb[X \ge xn]
				\le \psi \cdot 2^{-\!\Dkl(x \,{\|}\, p) \nwspace n}$.
	\end{enumerate}
\end{lemma}

\begin{proof}
	Statement~($\ref{clause:integral_case_above_mean}$) follows from 
	($\ref{clause:integral_case_below_mean}$) in the usual way
	since $\psi(n,p,x) = \varphi(n,1\,{-}\,p,1\,{-}\,x)$.
	Let $C = p \nwspace \sqrt{1-p}/16 \sqrt{2}$.
	Note that $C$ is at most $0.045$ for any $p$. 
	We first discuss the corner cases $x = 0$ and $x = p$ (assuming that $pn$ is an integer).
	If $x = 0$, then we have
	$\Pb[X \le 0 \cdot n] = (1-p)^n = \varphi(n,p,0) \cdot 2^{-\!\Dkl(0 \,{\|}\, p) \nwspace n}$.	
	If $x = p$,	the upper bound 
	$\Pb[X \le pn] \le 1 = \varphi(n,p,p) \cdot 2^{-\!\Dkl(p \,{\|}\, p) \nwspace n}$
	holds vacuously.
	The lower bound follows from $pn$ being the median of the binomial distribution,
	which implies $\Pb[X \le pn] \ge \sfrac{1}{2} 
	\ge C = C \cdot \varphi(n,p,p) \cdot 2^{-\!\Dkl(p \,{\|}\, p) \nwspace n}$.
	
	Assume $0 < x < p$.
	The original Chernoff--Hoeffding theorem and
	Lemma~\ref{lem:integral_bounded_away} together give
	\begin{equation*}
		\Pb[X \le xn] 
			\le \min\!\left(1, \,
				\frac{ p \nwspace \sqrt{1-x}}{(p - x) \nwspace \sqrt{\pi \nwspace xn}}\right)
					\cdot  2^{-\!\Dkl(x \,{\|}\, p) \nwspace n}
			\le \min\!\left(1, \,
				\frac{p}{\sqrt{\pi}}\,\frac{1}{(p - x) \nwspace \sqrt{\nwspace xn}}\right)
					\cdot  2^{-\!\Dkl(x \,{\|}\, p) \nwspace n}.
	\end{equation*}
	
	\noindent
	The latter is at most $\varphi \cdot 2^{-\!\Dkl(x \,{\|}\, p) \nwspace n}$.
	Finally, the lower bound in this case 
	is an easy consequence of Lemma~\ref{lem:integral_converge_to_p}
	and $p \nwspace \sqrt{1-x}/16 \sqrt{2}$ being larger than $C = p \nwspace \sqrt{1-p}/16 \sqrt{2}$.
\end{proof}

\subsection{General Case}
\label{subsec:general_case}

The second major step of the argument is to extend the result
above from integral products $xn$ to arbitrary real $x$.
The equality $\Pb[X \le xn] = \Pb[X \le \lfloor xn \rfloor \nwspace]$
holds universally as $X$ assumes only integer values.
In Section~\ref{subsec:integral_case}, we have given bounds
on the second probability $\Pb[X \le \lfloor xn \rfloor \nwspace]$
in terms of the ratio $\lfloor xn \rfloor/n$.
To reach the generality of the Chernoff--Hoeffding theorem,
we need to infer bounds on $\Pb[X \le xn]$ in terms of $x$.
In what follows, let $x'$ abbreviate $\lfloor xn \rfloor/n$.
%Observe that we have $x-\sfrac{1}{n} < x' \le x$.
Consider the upper bound in Lemma~\ref{lem:integral_case}~($\ref{clause:integral_case_below_mean}$)
as an illustrating example. 
It states that
\begin{equation*}
	\Pb[X \le xn] = \Pb[X \le x'n] 
		\le \min\!\left( 1, \,\frac{1}{(p-x') \nwspace \sqrt{x'n}} \right)
			\cdot 2^{-\!\Dkl(x' \,{\|}\, p) \nwspace n}.
\end{equation*}

\noindent
If there exists some constant $C$, possibly dependent on $p$ but independent of $x$ and $n$,
such that
\begin{equation*}
	\frac{1}{(p-x') \nwspace \sqrt{x'n}} \cdot 2^{-\!\Dkl(x' \,{\|}\, p) \nwspace n}
		\le \frac{C}{(p-x) \nwspace \sqrt{xn}} \cdot 2^{-\!\Dkl(x \,{\|}\, p) \nwspace n},
\end{equation*}

\noindent
then our estimate transfers to the general case.
A similar reasoning applies to the other cases.

The next two lemmas prepare the necessary technical machinery
to show the existence of those constants.
Lemma~\ref{lem:monotone_upper} clarifies the monotonicity of the functions in question.
It shows that transitioning from $x'$ to $x$ can only increase the upper bound,
meaning that we can actually choose $C = 1$ in the above illustration.
For the opposite direction,
Lemma~\ref{lem:linear_loss} asserts that this transitions
incurs a multiplicative loss that is at most linear in $x$.
Below, we often conclude the monotonic behavior
of a product of functions from that of its factors.
While in general the product of non-decreasing functions is not itself non-decreasing,
the monotonicity transfers if all factors are additionally non-negative.

\begin{lemma}
\label{lem:monotone_upper}
	Let $n$ be a positive integer, and $p$ and $x$ two probabilities.
	The function
	\begin{equation*}
		g_{n,p}(x) = \frac{2^{-\!\Dkl(x \,{\|}\, p) \nwspace n}}{(p-x) \nwspace \sqrt{xn}} 
	\end{equation*}
	is non-decreasing for all $x$ such that $\sfrac{1}{n} \le x < p$,
	provided that $n$ is sufficiently large.
\end{lemma}

\begin{proof}
	Quantity $2^{-\!\Dkl(x \,{\|}\, p) \nwspace n}$ is non-decreasing for $x \le p$
	(Lemma~\ref{lem:monotone_div})
	and it is not hard to prove this also for
	$1/(p-x)\sqrt{xn}$ given that $x \ge \sfrac{p}{3}$.
	The main focus of this proof is to show
	that the divergence power dominates the monotonicity of $g_{n,p}$ also
	for $\sfrac{1}{n} \le x \le \sfrac{p}{3}$.

	Taking derivatives gives
	\begin{align*}
	\hspace*{-.2em}													% avoid an overfull box
		\frac{\mathrm{d}}{\mathrm{d}x} \nwspace g_{n,p}(x)
			&= \frac{1}{(p-x)^2 \nwspace x \nwspace \sqrt{n}} 
				\left(\! \left(\frac{\partial}{\partial x} \nwspace 2^{-\!\Dkl(x \,{\|}\, p) \nwspace n}
				\right)\! (p-x) \sqrt{x} - 2^{-\!\Dkl(x \,{\|}\, p) \nwspace n} 
				\left(\frac{\partial}{\partial x} \nwspace (p-x) \sqrt{x} \right)\!\right)\\
			&= \frac{1}{(p-x)^2 \nwspace x \nwspace \sqrt{n}} 
				\left(\! n \ln\!\left( \frac{1-x}{x} \frac{p}{1-p} \right)
					2^{-\!\Dkl(x \,{\|}\, p) \nwspace n} \nwspace (p-x) \sqrt{x} 
					- 2^{-\!\Dkl(x \,{\|}\, p) \nwspace n} \nwspace
				\frac{p-3x}{2 \sqrt{x}}\right)\\[.25em]
			&= \frac{ 2^{-\!\Dkl(x \,{\|}\, p) \nwspace n}}
					{(p-x) \nwspace x^{\sfrac{3}{2}} \nwspace \sqrt{n}}
				\left(\!n \ln\!\left( \frac{1-x}{x} \frac{p}{1-p}
				\right)\! x - \frac{p-3x}{2 (p-x)}\right)\!.
	\end{align*}

	The first factor is positive for all $n$ and $x < p$
	and the same is true for the second one if $\sfrac{p}{3} < x$.
	Assume $x \le \sfrac{p}{3}$ in the remainder.
	Then, the last term of the second factor, $- (p-3x)/2 (p-x)$,
	is at least $- \nwspace\sfrac{1}{2}$.
	It is thus sufficient to prove the non-negativity of
	\begin{equation*}
		h_{n,p}(x) = n \ln\!\left(\frac{1-x}{x} \frac{p}{1-p} \right) \! x - \frac{1}{2}
	\end{equation*}
	
	\noindent
	on the subinterval $[\sfrac{1}{n},\, \sfrac{p}{3}]$ for all $n$ large enough.
	We do this in two claims.
	First, $h_{n,p}$ is concave there and,
	secondly, its values at the endpoints of the interval are non-negative.
	Regarding the concavity, observe that the derivative
	\begin{equation*}
		\frac{\mathrm{d}}{\mathrm{d}x} \, h_{n,p}(x) 
			= n \ln\!\left( \frac{1-x}{x} \frac{p}{1-p} \right) - \frac{n}{1-x}
	\end{equation*}
	
	\noindent
	is the sum of two non-increasing functions of $x$;
	$n \ln(\frac{1-x}{x} \frac{p}{1-p})$ is non-increasing
	as it is the derivative of the concave mapping $x \mapsto -\Dkl(x \,{\|}\, p) \nwspace n$.
	At the endpoint $\sfrac{1}{n}$, we have
	\begin{equation*}
		h_{n,p}\!\left(\frac{1}{n}\right) 
			= \ln\left( \!(n-1) \, \frac{p}{1-p} \right) - \frac{1}{2},
	\end{equation*}
	
	\noindent
	which is non-negative for all $n \ge (\sqrt{e} \nwspace (1-p)/p) + 1$.
	Similarly at endpoint $\sfrac{p}{3}$, 
	\begin{equation*}
		h_{n,p}\!\left(\frac{p}{3}\right) 
			= n \ln\left( \frac{3-p}{1-p} \right) \frac{p}{3} - \frac{1}{2}
	\end{equation*}
	
	\noindent
	is non-negative for $n \ge 3/ 2p \ln(\frac{3-p}{1-p})$.
\end{proof}

\begin{lemma}
\label{lem:linear_loss}
	Let $n$ be a positive integer, $p$ and $x$ two probabilities,
	$x' = \lfloor xn \rfloor/n$, and again
	\begin{equation*}
		g_{n,p}(x) = \frac{2^{-\!\Dkl(x \,{\|}\, p) \nwspace n}}{(p-x) \nwspace \sqrt{xn}}.
	\end{equation*}
	
	\noindent
	The following inequalities hold for all $x$ such that $\sfrac{1}{n} \le x < p$ 
	and $n$ sufficiently large.
	\begin{enumerate}[\textup{(}$i$\textup{)}]
		\item $g_{n,p}(x') \ge \frac{1-p}{p \sqrt{2}} 
			\nwspace e^{-\frac{1}{2-2p}} \cdot x \cdot g_{n,p}(x)$;
		\item $2^{-\!\Dkl(x' \,{\|}\, p) \nwspace n} 
		\ge \frac{1-p}{p} \nwspace  e^{-\frac{1}{2-2p}} \cdot x 
			\cdot 2^{-\!\Dkl(x \,{\|}\, p) \nwspace n}$.
	\end{enumerate}
\end{lemma}

\begin{proof}
	We make heavy use of the facts that $x - \sfrac{1}{n} < x' \le x$,
	and that $x \ge \sfrac{1}{n}$ implies $x' \ge \sfrac{1}{n}$.
	The relative difference between $g_{n,p}(x')$ and $g_{n,p}(x)$ is
	\begin{equation*}
		\frac{g_{n,p}(x')}{g_{n,p}(x)} = 
			\frac{2^{-\!\Dkl(x' \,{\|}\, p) \nwspace n}}{2^{-\!\Dkl(x \,{\|}\, p) \nwspace n}}
				\cdot \frac{p-x \ }{p-x'} \cdot \sqrt{\frac{x}{x'}}.
	\end{equation*}
	
	\noindent
	The last factor is at least $1$ and the middle one is
	$\frac{p-x\,\nwspace}{p-x'} = 1 - \frac{x-x'}{p-x'} \ge 1 - \frac{1}{pn} \ge \frac{1}{\sqrt{2}}$
	given that $n \ge 2/(2-\sqrt{2})p$.
	The first factor can be estimated using Lemma~\ref{lem:monotone_div} and $x'-x \le \sfrac{1}{n}$ as
	\begin{align*}
		\frac{2^{-\!\Dkl(x' \,{\|}\, p) \nwspace n}}{2^{-\!\Dkl(x \,{\|}\, p) \nwspace n}}
			&= \left(\frac{x}{1-x} \, \frac{1-p}{p}\right)^{(x-x')n}\!
				2^{-\!\Dkl(x' \,{\|}\, x) \nwspace n} \cdot 2^{-\!\Dkl(x \,{\|}\, p) \nwspace n}\\
			&\ge \frac{x}{1-x} \, \frac{1-p}{p} \
				2^{-\!\Dkl(x' \,{\|}\, x) \nwspace n} \cdot 2^{-\!\Dkl(x \,{\|}\, p) \nwspace n}
			 \ge x \cdot \frac{1-p}{p} \ 2^{-\!\Dkl(x' \,{\|}\, x) \nwspace n}
			 	\cdot 2^{-\!\Dkl(x \,{\|}\, p) \nwspace n}.
	\end{align*}
	
	\noindent
	It remains to show that $2^{-\!\Dkl(x' \,{\|}\, x) \nwspace n}$ is at least a constant,
	namely, we claim $2^{-\!\Dkl(x' \,{\|}\, x) \nwspace n} > e^{-\frac{1}{2-2p}}$.
	Let $t^- = \argmin_{t \in [x',x]} t(1-t)$,
	observe that $\sfrac{1}{n} \le t^- < p$ holds.
	By Lemma~\ref{lema:div_quadratic}, the exponent $\Dkl(x' \,{\|}\, x) \nwspace n$ (to the base $\sfrac{1}{2}$) is bounded.
	\begin{equation*}
		\Dkl(x' \,{\|}\, x) \nwspace n \le \frac{(x-x')^2}{t^- (1-t^-) \nwspace 2 \ln 2} \cdot n
			< \frac{\frac{1}{n^2}}{\frac{1}{n} (1-t^-) \nwspace 2 \ln 2} \cdot n 
			= \frac{1}{(1-t^-) \nwspace 2 \ln 2}
			< \frac{1}{(2-2p) \ln 2}.\qedhere
	\end{equation*}
\end{proof}

We have the tools ready to prove Theorem~\ref{thm:better_Hoeffding} in its entirety.

\begin{backInTime}{thm:better_Hoeffding}
\begin{theorem}[restated with explicit constants]
	Let $n$ be a positive integer, $0 < p < 1$ a non-trivial probability, and
	$X  \sim \Bin(n,p)$ a binomial variable.
	Suppose the function $x = x(n)$ takes real values in the interval
	$[\varepsilon, \, 1\,{-}\,\varepsilon]$ for some $\varepsilon > 0$.
	Let $\varphi$ and $\psi$ denote the functions
	\begin{equation*}
		\varphi(n,p,x) = \min\!\left( 1, \,\frac{1}{(p-x) \nwspace \sqrt{xn}} \right)
		\ \, \text{and}\quad
		\psi(n,p,x) = \min\!\left( 1, \,\frac{1}{(x-p) \nwspace \sqrt{(1-x)n}} \right)\!,
	\end{equation*}
	
	\noindent
	with additionally $\varphi(n,p,p) = 1$ and $\psi(n,p,p) = 1$.
	The following statements hold for all $n$ sufficiently large.

	\begin{enumerate}[\textup{(}$i$\textup{)}]
		\item\label{case:restated_Chernoff-Hoeffding} If $x \le p$, then 
				$\frac{\varepsilon \nwspace (1-p)^{\frac{3}{2}}}{32} \nwspace e^{-\frac{1}{2-2p}}
					\cdot \varphi \cdot 2^{-\!\Dkl(x \,{\|}\, p) \nwspace n} \le \Pb[X \le xn]
					\le \varphi \cdot 2^{-\!\Dkl(x \,{\|}\, p) \nwspace n}$.
		\item If $x \ge p$, then
				$\frac{\varepsilon \nwspace p^{\frac{3}{2}}}{32} \nwspace e^{-\frac{1}{2p}} 
					\cdot \psi \cdot 2^{-\!\Dkl(x \,{\|}\, p) \nwspace n} \le \Pb[X \ge xn]
					\le \psi \cdot 2^{-\!\Dkl(x \,{\|}\, p) \nwspace n}$.
	\end{enumerate}
\end{theorem}
\end{backInTime}

\begin{proof}
	We only need to prove the first statement.
	Let $x' = \lfloor xn \rfloor/n$.
	This implies $x' \le x$ and
	makes $x'n$ an integer such that $\Pb[X \le x'n] = \Pb[X \le xn]$.
	Recall the definition of $g_{n,p}$ from Lemma~\ref{lem:monotone_upper}.
	It is chosen such that for all $n$, $p$, and $x$, we have
	\begin{equation*}
		\varphi(n,p,x') \cdot 2^{-\!\Dkl(x' \,{\|}\, p) \nwspace n} =
			\begin{cases}
				2^{-\!\Dkl(x' \,{\|}\, p) \nwspace n}, & \text{if } 1 \le \frac{1}{(p-x') \sqrt{x'n}}
					\text{ or } x' = p;\\
				g_{n,p}(x'), & \text{otherwise.}
			\end{cases}
	\end{equation*}
	
	Lemmas~\ref{lem:monotone_div} and \ref{lem:monotone_upper}
	together establish
	$\varphi(n,p,x') \cdot 2^{-\!\Dkl(x' \,{\|}\, p) \nwspace n}
		\le \varphi(n,p,x) \cdot 2^{-\!\Dkl(x \,{\|}\, p) \nwspace n}$
	in both cases, provided that $n$ is large enough.
	The upper bound in Statement~($\ref{case:restated_Chernoff-Hoeffding}$) now
	follows from the integral case in Lemma~\ref{lem:integral_case}.
	Regarding the lower bound, Lemma~\ref{lem:linear_loss} gives
	for all $n$ large enough,
	\begin{equation*}
		\varphi(n,p,x') \cdot 2^{-\!\Dkl(x' \,{\|}\, p) \nwspace n} \ge
			\begin{cases}
				\frac{1-p}{p} \nwspace e^{-\frac{1}{2-2p}} \cdot x
					 \cdot 2^{-\!\Dkl(x \,{\|}\, p) \nwspace n},
					 	& \text{if } 1 \le \frac{1}{(p-x') \sqrt{x'n}} \text{ or } x' = p;\\
				\frac{1-p}{p \sqrt{2}} \nwspace e^{-\frac{1}{2-2p}} \cdot x \cdot g_{n,p}(x),
						& \text{otherwise.}
			\end{cases}
	\end{equation*}
	
	\noindent
	In summary, using Lemma~\ref{lem:integral_case} and the assumption $x \ge \varepsilon$, 
	we have 
	\begin{align*}
		\Pb[X \le xn] = \Pb[X \le x'n]
			&\ge \frac{p \nwspace \sqrt{1-p}}{16 \sqrt{2}} 
				\cdot \varphi(n,p,x') \cdot 2^{-\!\Dkl(x' \,{\|}\, p) \nwspace n}\\ 
			&\ge \frac{\varepsilon \nwspace (1-p)^{\frac{3}{2}}}{32} \nwspace e^{-\frac{1}{2-2p}} 
				\cdot \varphi(n,p,x) \cdot 2^{-\!\Dkl(x \,{\|}\, p) \nwspace n}.\ \qedhere
	\end{align*}
\end{proof}

\subsection{Applications}
\label{subsec:applications}

In this excursive section, we highlight two applications of Theorem~\ref{thm:better_Hoeffding}
that do not fall into our main objective of studying hypergraphs,
we find them instructive nevertheless.
First, we give an alternative proof for the intuition
that the probability of a binomial variable taking values outside of a few standard variation
of its expectation may be small but does not converge to $0$.
The result was previously obtained via the Berry--Esseen inequality
using the normal approximation of the binomial distribution~\cite{Oliveto15ImprovedTimeComplexity}.
Secondly, we improve the rate of convergence in Cramér's theorem.

\vspace*{-.75em}
\subparagraph*{Anti-Concentration Inequalities.}
Chernoff bounds are usually interpreted as concentration inequalities~\cite{Dubhashi09ConcentrationOfMeasure}.
They state that the mass of the binomial distribution is concentrated around its mean
and the probability for any other value falls exponentially in the distance to the expectation.
However, it is also known that the probability of a binomial variable $X \sim \Bin(n,p)$
taking values outside of a constant number of standard deviations from $\Ev[X]$ does not vanish,
even as $n$ grows large.
Results of the latter kind are occasionally called anti-concentration inequalities.
The particular bound\footnote{%
	Lemma~6.1 in~\cite{Oliveto15ImprovedTimeComplexity} only states the existence of
	a particular pair of constants $c$, $C$.
	It is easy to check that their proof remains valid for any $c \ge 0$.
	For a generalization to unequal probabilities,
	see Lemma~1.10.16 in~\cite{Doerr20Evolutionary}.
}
we are interested in was given by Oliveto and Witt~\cite{Oliveto15ImprovedTimeComplexity}.
We give an alternative proof that avoids the normal distribution.

\begin{corollary}[Lemma~6.1 in~\cite{Oliveto15ImprovedTimeComplexity}]
	Let $n$ be a positive integer, $0< p < 1$ a non-trivial probability,
	and $X \sim \Bin(n,p)$ a binomial variable.
	Let $\sigma_X = \sqrt{np(1-p)}$ denote the standard deviation of $X$.
	For every non-negative real $c \ge 0$, there exists some positive $C > 0$,
	independent of $n$ but possibly dependent on $c$ and $p$,
	such that
	\begin{equation*}
		\Pb[X \le \Ev[X] - c \cdot \sigma_X] \ge C.
	\end{equation*}
	
	\noindent
	Conversely, for any non-negative function $f = f(n)$ with $\lim_{n \to \infty} f(n) =  \infty$,
	it holds that
	\begin{equation*}
		\lim_{n \to \infty} \Pb[X \le \Ev[X] - f \cdot \sigma_X] = 0.
	\end{equation*}
\end{corollary}

\begin{proof}
	Let $c'= c \sqrt{p(1-p)}$.
	In the notation of the previous sections, we have
	\begin{equation*}
		x = \frac{\Ev[X] - c \cdot \sigma_X}{n} 
	%		= \frac{pn - c \cdot \sqrt{p(1-p)\nwspace n}}{n} 
			= p - \frac{c'}{\sqrt{n}}.
	\end{equation*}
	
	\noindent
	We assume $x \ge \sfrac{p}{2}$,
	this does not loose generality as $x$ converges to $p$.
	Applying Theorem~\ref{thm:better_Hoeffding} gives
	\begin{equation*}
		\Pb[X \le \Ev[X] - c \cdot \sigma_X] 
			\ge \frac{p \nwspace (1-p)^{\frac{3}{2}}}{64} e^{-\frac{1}{2-2p}}
				\cdot \min\!\left(1, \, \frac{1}{c' \nwspace \sqrt{x}}\right) 
				\cdot 2^{-\!\Dkl(x \,{\|}\, p) \nwspace n}.
	\end{equation*}
	
	\noindent
	The minimum is at least $1/cp\sqrt{1-p}$.
	To show that also the last factor is bounded, we use an argument
	very similar to the one in Lemma~\ref{lem:linear_loss}.
	Let $t^- = \argmin_{t \in [x,p]} t(1-t)$.
	We have $\sfrac{p}{2} \le t^- \le p$, resulting in
	\begin{equation*}
		\Dkl(x \,{\|}\, p) \cdot n 
		\le \frac{\left(\frac{c'}{\sqrt{n}}\right)^2}{t^- (1-t^-) \nwspace 2 \ln 2} \cdot n
			= \frac{c^2 \, p (1-p)}{t^-(1-t^-) \nwspace 2 \ln 2}
			\le \frac{c^2 \, p (1-p)}{\frac{p}{2}(1-p) \nwspace 2 \ln 2}
			= \frac{c^2}{\ln 2}.
	\end{equation*}
	
	If we move away from $\Ev[X]$ by more than a constant number of standard deviations,
	already the original Chernoff--Hoeffding theorem is sharp enough to prove 
	that the distribution function vanishes.
	For some non-negative $f(n) = \littleOm(1)$, 
	define $x = p - (f \sqrt{p(1-p)} \nwspace/\sqrt{n} \nwspace)$.
	If $x$ is negative, the statement holds vacuously;
	otherwise, we have $\Pb[X \le \Ev[X] - f \cdot \sigma_X] \le 2^{-\!\Dkl(x \,{\|}\, p) \nwspace n}$.
	Let $t^+ = \argmax_{t \in [x,p]} t(1-t)$.
	Lemma~\ref{lema:div_quadratic} gives
	\begin{equation*}
	 	\Dkl(x \,{\|}\, p) \cdot n
	 		\ge \frac{f(n)^2 \ p(1-p)}{t^+(1-t^+) \nwspace 2 \ln 2} = \littleOm(1). \qedhere
	\end{equation*} 
\end{proof}

\vspace*{-.75em}
\subparagraph*{Cramér's Theorem.}
Theorem~\ref{thm:better_Hoeffding} also has implications 
for the rate of convergence in Cramér's theorem in the theory of large deviations.
Fix some non-trivial probability $p$ and let $(X_i)_i$ be a sequence of i.i.d.\ 
Bernoulli variables\footnote{%
	Cram{\'e}r's theorem holds more generally for any i.i.d.\ sequence $(X_i)_i$
	such that the cumulant-generating function of $X_1$ is finite everywhere~\cite{Deuschel84LargeDeviations}.
}
with success probability $p$.
Define
\begin{equation*}
	\Lambda^*(x) =  \sup_{t \in \mathbb{R}} 
		\nwspace \Big\{ tx - \ln ( \nwspace \Ev[e^{t X_1}] \nwspace) \Big\}
\end{equation*}

\noindent
to be the Legendre transform of the cumulant-generating function of $X_1$.
Cram{\'e}r's theorem~\cite{Cramer38SurUn,Deuschel84LargeDeviations} states that the transform observes the following
limiting property for all $x$ with $p < x < 1$,
\begin{equation*}
	\lim_{n \to \infty}\  \frac{1}{n} \nwspace 
		\ln\!\left( \Pb\!\left[ \sum_{i=1}^n X_i \ge xn \right] \right)
		= -\Lambda^*(x).
\end{equation*}

\noindent
For notational convenience, let
$\Dkl_e(x \,{\|}\, y) = \ln(2) \cdot \Dkl(x \,{\|}\, y)
	= -x \nwspace \ln\!\left( \frac{y}{x} \right) - (1-x) 
		\nwspace \ln\!\left(\frac{1-y}{1-x} \right)$
denote the natural (base-$e$) Kullback--Leibler divergence.
It is straightforward to verify from $\Ev[e^{t X_1}] = 1-p + p \nwspace e^t$
that Cram{\'e}r's function $\Lambda^*(x) = \Dkl_e(x \,{\|}\, p)$
is in fact the natural divergence.
This shows that the original Chernoff--Hoeffding inequality with
$2^{-\!\Dkl(x \,{\|}\, p) \nwspace n} = e^{-\!\Dkl_e(x \,{\|}\, p) \nwspace n}$
is asymptotically tight up to
sublinear terms in the exponent.
However, the rate of convergence of the above limit is subject of ongoing
research~\cite{Deuschel84LargeDeviations,Dembo10LargeDeviations,Fill83ConvrgenceRates}.
Observe that under the current assumptions,
we have
$1 \ge 1/(x-p) \nwspace \sqrt{(1-x)n}$
for all $n$ sufficiently large.
Theorem~\ref{thm:better_Hoeffding} thus implies the following corollary, 
which further clarifies the rate in the Bernoulli case.

\begin{corollary}
\label{cor:Cramer}
	Let $0< p < 1$ be a non-trivial probability 
	and $(X_i)_{i}$ a sequence of i.i.d.\ Bernoulli variables with parameter $p$.
	Then, for any $x$ with $p < x < 1$, it holds that
	\begin{equation*}
		\frac{1}{n} \ln\!\left( \Pb\!\left[ \sum_{i=1}^n X_i \ge xn \right] \right)
			= - \Dkl_e(x \,{\|}\, p) - \frac{1}{2} \nwspace \frac{\ln n}{n} 
				- \frac{\ln(x-p)}{n} - \frac{1}{2} \nwspace \frac{\ln(1-x)}{n}
				\pm \Or\!\left(\frac{1}{n}\right)\!. 
	\end{equation*}
\end{corollary}

\pagebreak

\section{Distinct Sets and Minimality}
\label{sec:distinct_sets_and_minimality}

We return to the main topic of this work,
which is determining the expected number of minimal edges
in the maximum-entropy multi-hypergraph $\B_{n,m,p}$.
The vertex sampling probabilities $p=0$ or $p=1$ result in trivial hypergraphs,
we thus assume $0 < p < 1$ in the remainder of this work,
unless explicitly stated otherwise.
Also $n$, $m$ always denote positive integers and $X \sim \Bin(n,p)$ a binomial variable.
Under this assumptions, every subset of $[n]$ then has a non-vanishing chance to be sampled.
Such a set is minimal for $\B_{n,m,p}$
if and only if it is generated in one of the trials
and no proper subset ever occurs.
Both of these aspects influence the chance of minimality,
but their impact varies depending on the cardinality of the sets in question.

The number of vertices per edge is heavily concentrated around $pn$
and the more vertices there are in an edge, the less likely it is minimal.
Intuitively speaking, almost no sets with very low cardinalities are sampled,
but if so, they are very likely included in the minimization.
There are plenty of edges with a medium number of vertices
and there is still a good chance they are minimal.
Sets of high cardinality rarely occur and usually
they are dominated by smaller edges.
This disparity is exacerbated by a large number of trials.
Boosting $m$ increases the probability that also sets of cardinality
a bit further away from $pn$ are sampled,
at the same time the process now generates more duplicate sets that do not count towards the minimization.
More importantly though, the likelihood of a larger set being minimal
is smaller with many trials.
Eventually, the last effect outweighs all others,
creating the situation in which the only minimal edge is empty.

In this section, we make this intuition rigorous.
We start by giving preliminary bounds on the number of minimal edges
as a first step towards the proof of Theorem~\ref{thm:binomial_characterization}.
These bounds have the form of are binomial sums of polynomials of probabilities,
depending on which factors we include in those sums, we get an upper or a lower bound.
The estimates are already tight up to constants but are rather unwieldy.
They will serve as the basis for our further analysis.

Let $\mathcal{D}_{n,p}$ denote the maximum-entropy distribution on the power set $\mathcal{P}([n])$
provided that $\Ev_{S \sim \mathcal{D}_{n,p}}[ \nwspace | S | \nwspace] = pn$,
each vertex is included independently with probability $p$.
Let $S_j \sim \mathcal{D}_{n,p}$ be the outcome of the $j$-th trial.
Define $s_{n,p}(i,m) = \Pb[\nwspace \exists j \le m \colon  S_{j} = [i] \nwspace]$
to be the probability that the set $[i]$ 
(in fact, any set of cardinality $i$) is sampled,
and $w_{n,p}(i,m) = \Pb[\nwspace \forall j \le m \colon \neg (S_{j} \subsetneq [i]) \nwspace]$
as the probability that no trial ever samples a proper subset of $[i]$.\footnote{%
	The notation $s_{n,p}$ refers to the set being \emph{sampled},
	these probabilities are then \emph{weighted} by the factors $w_{n,p}$.
}

\begin{lemma}
\label{lem:formula_expt_Bnk}
	We have $s_{n,p}(i,m) = 1 - (1 - p^i (1-p)^{n-i})^m$ and
	$w_{n,p}(i,m) = (1 -(1-p)^{n-i}(1-p^i))^m$.
	The following statements hold for the minimization of multi-hypergraph $\B_{n,m,p}$.
	\begin{enumerate}[\textup{(}$i$\textup{)}]
		\item\label{clause:lower_bound} $\Ev[ \nwspace |\min(\B_{n,m,p})|\nwspace ]
			\ge \sum_{i=0}^n \binom{n}{i} \nwspace s_{n,p}(i,m) \cdot w_{n,p}(i,m)$.
		\item\label{clause:upper_bound} $\Ev[ \nwspace |\min(\B_{n,m,p})| \nwspace ]
			\le \sum_{i=0}^n \binom{n}{i} \nwspace s_{n,p}(i,m) \cdot w_{n,p}(i, m -1)$.
		\item\label{clause:constant_factor} $\Ev[ \nwspace |\min(\B_{n,m,p})| \nwspace ]
			\le 1 + \frac{1}{p} \sum_{i=0}^n \binom{n}{i} \nwspace s_{n,p}(i,m) \cdot w_{n,p}(i,m)$.	
	\end{enumerate}
\end{lemma}

\begin{proof}
	To see the closed forms of $s_{n,p}(i,m)$ and $w_{n,p}(i,m)$,
	observe that the random set \mbox{$S_j  \sim \mathcal{D}_{n,p}$}
	differs from $[i]$ with probability $1-p^i(1-p)^{n-i}$.
	The independence of the trials gives $s_{n,p} = 1 - (1 - p^i (1-p)^{n-i})^m$.
	Similarly, $S_j$ is a subset of $[i]$
	if it does not contain an element of $[n]{\setminus}[i]$,
	having probability $(1-p)^{n-i}$,
	Conditioned on being any subset, $S_j$ is a \emph{proper} subset
	if it is missing at least one element of $[i]$.
	The expression for $w_{n,p}(i,m)$ follows from here.
	
	Regarding the main statements,
	some fixed set $S \subseteq [n]$ is in $\min(\B_{n,m,p})$
	if and only if it is sampled in one of the $m$ trials and no proper subset is sampled.
	The probability for both events depends only on $|S|$
	as all sets with the same cardinality are equally likely.
	\begin{gather*}
		\Ev[ \nwspace |\min (\B_{n,m,p})| \nwspace ]
			= \sum_{S \subseteq [n]} \Pb[ \nwspace (\exists k \le m \colon S_k = S)
				\wedge (\forall j \le m \colon \neg(S_{j} \subsetneq S))\nwspace ]\\
		\quad\quad= \sum_{i=0}^n \binom{n}{i} \Pb[ \nwspace \exists k \le m \colon S_k = [i] \nwspace]
				\cdot \Pb[ \nwspace \forall j \le m \colon \neg(S_{j} \subsetneq [i]) 
				\mid \exists k \le m \colon S_k = [i] \nwspace]\\[.5em]
		\quad\quad= \sum_{i=0}^n \binom{n}{i} \, s_{n,p}(i,m)
				\cdot \Pb[ \nwspace \forall j \le m \colon \neg(S_{j} \subsetneq [i]) 
				\mid \exists k \le m \colon S_k = [i] \nwspace].
	\end{gather*}

	The last factor describes the likelihood that any set with $i$ elements is minimal,
	conditioned on it being sampled at all.
	The stated bounds differ only in the way this factor is estimated.
	We claim it to be at least
	$\Pb[ \nwspace \forall j \le m \colon \neg(S_{j} \subsetneq [i]) \nwspace ]$
	(that is, without the condition) while at the same time being at most
	$\Pb[ \nwspace \forall j < m \colon \neg(S_{j} \subsetneq [i]) \nwspace ]$
	(with one fewer trial).
	The first inequality is obvious because conditioning on some trial 
	producing $[i]$ itself only increases the chances of never sampling a proper subset.
	For the second one, we apply Lemma~\ref{lem:znZN}
	to the events $A_j = [S_j \subsetneq [i] \nwspace]$ and $B_j = [S_j = [i] \nwspace]$,
	which shows that 
	\begin{equation*}
		\Pb[ \nwspace \forall j \le m \colon \neg(S_{j} \subsetneq [i]) 
			\mid \exists k \le m \colon S_k = [i] \nwspace ] 
			\le \Pb[ \nwspace \forall j \le m \colon \neg(S_{j} \subsetneq [i])
				\mid S_m = [i] \nwspace ].
	\end{equation*}
	
	\noindent
	The proof of the claim, and thereby the one of Statements~($\ref{clause:lower_bound}$)
	and ($\ref{clause:upper_bound}$), is completed by the fact that 
	$\Pb[ \nwspace \forall j \le m \colon \neg(S_{j} \subsetneq [i]) \mid S_m = [i] \nwspace ]$
	is equal to $\Pb[ \nwspace \forall j < m \colon \neg(S_{j} \subsetneq [i]) \nwspace ]$.

	In order to prove Statement~($\ref{clause:constant_factor}$),
	note that the $i$-th term of the sums in the first two statements
	have a relative difference of $w_{n,p}(i,m)/w_{n,p}(i,m{-}1)$.
	By independence, this is equal to $w_{n,p}(i,1)$,
	the probability to not sample a strict subset of $[i]$ in a single trial.
	For $i <n$, it is easy to see that $w_{n,p}(i,1) \ge p$.
	If $i = n$, we have $w_{n,p}(n,1) = p^n$ which is sub-constant.
	However, the statement follows anyway
	as the contribution of the last term to the whole sum is at most $1$.
\end{proof}

The part that all three bounds of Lemma~\ref{lem:formula_expt_Bnk} have in common describes
the expected number of \emph{distinct} sets in $\B_{n,m,p}$.
Recall that we use $\|\Hyp\|$ to denote the number of distinct sets of some multi-hypergraph $\Hyp$.
That means, we have
\begin{equation*}
	\Ev[ \nwspace \|\B_{n,m,p}\| \nwspace ]
		= \sum_{i=0}^n \binom{n}{i} \, s_{n,p}(i,m)
		= \sum_{i=0}^n \binom{n}{i} \Big(1 - (1 - p^i (1-p)^{n-i})^m \Big).		
\end{equation*}

\noindent
We weight the terms of the sum by $w_{n,p}(i,m)$ or $w_{n,p}(i,m {-} 1)$
to estimate the size of the minimization.
We analyze the two parts separately,
starting with the \emph{weighting factors} $w_{n,p}$.

The behavior of the $w_{n,p}$ may have applications besides our study of random hypergraphs.
Consider $m$ trials according to the maximum-entropy distribution $\mathcal{D}_{n,p}$
on subsets of $[n]$.
Then, $w_{n,p}(i,m)$ is by definition the probability that
any fixed subset of cardinality $i$ survives as minimal after $m$ trials,
equivalently, any proper subset is sampled with probability $1- w_{n,p}(i,m)$.
It is easy to see that
%$w_{n,p}(i,m) = (1 -(1-p)^{n-i} + p^i(1-p)^{n-i}\nwspace)^m$
weighting factors are non-increasing in both $i$ and $m$.
We prove next that the weighting factors are in fact threshold functions falling abruptly
from almost $1$ to almost $0$ as $i$ increases from $0$ to $n$,
the position of the transition depends on $n$, $m$, and $p$.
Recall that $\alpha = -(\log_{1-p} m)/n$.
Lemma~\ref{lem:threshold_fct} below establishes
a sharp threshold behavior of $w_{n,p}(i,m)$ at
\begin{equation*}
	i^* = n + \log_{1-p}m = (1-\alpha) \nwspace n.
\end{equation*}

\noindent
Note that $i^*$ is always at most $n$ since $\log_{1-p} m$ is non-positive.
The definition is such that it ensures the equality $m = 1/(1-p)^{n-i^*} = 1/(1-p)^{\alpha n}$.
For increasing $m$, the threshold gets smaller relative to $n$.
Once $m$ grows beyond $1/(1-p)^n$, i.e., $\alpha > 1$,
the quantity $i^*$ can no longer be interpreted as a cardinality as it becomes negative.
Later, in Lemma~\ref{lem:char_trivial}, we will see that $m$ being this large 
is in fact irrelevant for the analysis of the minimization.

\pagebreak

\begin{lemma}
\label{lem:threshold_fct}
	It holds that \mbox{$w_{n,p}(0,m) = 1$} and $w_{n,p}(n,m) = p^{nm}$.
	Suppose $i = i(n)$ takes integer values such that $0 < i < n$, then
	\begin{enumerate}[\textup{(}$i$\textup{)}]
		\item\label{clause:exact_bounds}
			$\exp(- m \nwspace (1-p)^{n-i}) \cdot (1- m \nwspace (1-p)^{2(n-i)}) \le w_{n,p}(i,m)
			\le \exp(- m \nwspace (1-p)^{n-i+1})$.
	\end{enumerate}

	\noindent
	In particular, the following statements hold. %\footnote{%
%		We understand $\littleOm(1)$ as the class of all non-negative unbounded functions of $n$.
%		In particular, the classes $n + \log_{1-p} m + \littleOm(1)$ 
%		and $n + \log_{1-p} m - \littleOm(1)$ are disjoint.
%		}

	\begin{enumerate}[\textup{(}$i$\textup{)}]
	\setcounter{enumi}{1}
		\item\label{clause:above_threshold} If $i = n + \log_{1-p} m + \littleOm(1)$,
			then $\lim_{n \to \infty} w_{n,p}(i,m) = 0$.
		\item\label{clause:below_threshold} If $i = n + \log_{1-p} m - \littleOm(1)$,
			then $\lim_{n \to \infty} w_{n,p}(i,m) = 1$.
		\item\label{clause:at_threshold} If $i = n + \log_{1-p} m \pm \Th(1)$,
			then $w_{n,p}(i,m) = \Th(1)$.
	\end{enumerate}
\end{lemma}

\begin{proof}
	The corner cases are elementary.
	Assume $0 < i < n$ for the rest of the proof.
	We estimate $w_{n,p}(i,m)$ using mainly Corollary~\ref{cor:tight_bound_polynom}.
	This yields
	\begin{equation*}
		w_{n,p}(i,m) = (1- (1-p)^{n-i}(1-p^i))^{m}
			 \le(1- (1-p)^{n-i} (1-p))^m
			\le  \exp\!\Big(\!- m \nwspace (1-p)^{n-i} \cdot (1-p)\Big).
	\end{equation*}
	
	\noindent
	The limiting behavior is determined
	by the product $m \nwspace (1-p)^{n-i}\!$.
	If $i = n + \log_{1-p} m + \littleOm(1)$, then
	$m \nwspace (1-p)^{n-i}
		= m \nwspace (1-p)^{- (\log_{1-p} m) - \littleOm(1)}
		= (1-p)^{-\omega(1)}$	
	diverges and the weighting factor $w_{n,p}(i,m)$ converges to $0$.
	Conversely, we get from $1-p^i \le 1$ that
	\begin{equation*}
		w_{n,p}(i,m) \ge (1- (1-p)^{n-i})^m 
		\ge \exp\!\Big(\!- m \nwspace (1-p)^{n-i} \Big) \cdot (1- m \nwspace (1-p)^{2(n-i)}).
	\end{equation*}
	
	\noindent
	If $i = n + \log_{1-p} m - \littleOm(1)$, both $m \nwspace (1-p)^{n-i} = (1-p)^{\littleOm(1)}$
	and $m \nwspace (1-p)^{2(n-i)} = (1-p)^{\littleOm(1)}/m$ tend to $0$,
	implying $\lim_{n \to \infty} w_{n,p}(i,m) = 1$.
	
	Finally, if the cardinality $i$ is close to the threshold $i^* = n + \log_{1-p} m$,
	the limit may not exist. We show that $w_{n,p}(i,m)$ is still bounded away from $0$ for all $n$.
	Suppose $i = n + \log_{1-p} m \pm \Th(1)$, in particular, the difference $i^*-i$ is bounded.
	If $m$ is constant with respect to $n$, so is $w_{n,p}(i,m) \ge (1- (1-p)^{n-i})^m \ge p^m$.
	Here, we used the assumption $i < n$.
	If $m$ diverges, then $n-i = \log_{1-p} m \mp \Th(1) = \littleOm(1)$ diverges with it
	and thus
	\begin{align*}
		w_{n,p}(i,m)
			&\ge \exp\!\Big(\!- m \nwspace (1-p)^{n-i} \Big) (1- m \nwspace (1-p)^{2(n-i)})\\
			&= \exp\!\Big(\!- (1-p)^{i^*-i} \Big) \cdot (1- (1-p)^{ (i^*-i) + (n-i)})
			 = \Om(1).\qedhere
	\end{align*}
\end{proof}

After demonstrating a sharp threshold for the weighting factors,
we turn to the number of distinct sets $\|\B_{n,m,p}\|$.
%This will later serve as a natural upper bound for the size of the minimization.
A trivial cap for distinct sets is the total number of edges $m$.
When starting the sampling, many different sets are generated
and $\|\B_{n,m,p}\|$ is close to $m$.
As the number of trials increases though, duplicates occur in the sample
and the two quantities grow apart.
To discuss this in more detail, we introduce some notation.
For a pair of integers $\ell,\,u$ with $0 \le \ell \le u \le n$,
let $\| \B_{n,m,p}(\ell,u) \|$ denote the number of distinct samples whose cardinality is
between $\ell$ and $u$, including.
This is also at most as large as the total number of samples in that range.
It thus makes sense to expect an upper bound in terms of the binomial distribution.
We confirm this below and further prove that there is also a lower bound of the same flavor.

\begin{lemma}
\label{lem:distinct_sets}
	Let $\ell,\,u$ be integers such that $0 \le \ell \le u \le n$ and
	$\p = \max_{\ell \le i \le u} \nwspace \{ \nwspace p^i (1-p)^{n-i} \nwspace \}$,
	then
	\begin{equation*}
		\p = \begin{cases}
				p^\ell (1-p)^{n- \ell}, &\text{if } p<1/2;\\
				1/2^n,					&\text{if } p=1/2;\\
				p^u (1-p)^{n- u},		&\text{otherwise.}
			\end{cases}
	\end{equation*}

	\noindent
	The expected number of distinct sets in $\B_{n,m,p}$
	with cardinality between $\ell$ and $u$ observes
	\begin{equation*}
		\frac{m}{1+ m \p} \cdot \Pb[ \nwspace \ell \le X \le u \nwspace ]
			\le \Ev[ \nwspace \| \B_{n,m,p}(\ell,u) \| \nwspace ]
			\le m \cdot \Pb[ \nwspace \ell \le X \le u \nwspace ].
	\end{equation*}
\end{lemma}

\begin{proof}
	The closed form for $\p$ can be seen from the equality
	$p^i(1-p)^{n-i} = (p/(1-p))^i \cdot (1-p)^n$
	since the odds $p/(1-p)$ are strictly smaller than $1$
	iff $p < 1/2$.
	For the number of distinct sets in $\B_{n,m,p}(\ell,u)$,
	Lemma~\ref{lem:Badkobeh_bound} implies that
	\begin{equation*}
		\Ev[\nwspace \| \B_{n,m,p}(\ell,u) \| \nwspace] 
			= \sum_{i=\ell}^u \binom{n}{i} (1 - (1 - p^i (1-p)^{n-i})^m )
			\le m \cdot \sum_{i= \ell}^u \binom{n}{i} p^i (1-p)^{n-i}.
	\end{equation*}
	
	\noindent
	Conversely, we have
	\begin{equation*}
		\Ev[\nwspace \| \B_{n,m,p}(\ell,u) \| \nwspace] \ge \sum_{i= \ell}^u \binom{n}{i} \frac{m \nwspace p^i (1-p)^{n-i}}{1+ m \nwspace p^i (1-p)^{n-i}}
			\ge \frac{m}{1+ m \p} \cdot \sum_{i= \ell}^u \binom{n}{i} p^i (1-p)^{n-i}. \qedhere
	\end{equation*}
\end{proof}

\section{Size of the Minimization}
\label{sec:proof_of_main_thm}

We now prove the main theorems with the help of the tools above.
The key observation is that the minimization is dominated by the sets with cardinalities 
around $i^* = n + \log_{1-p} m = (1-\alpha) n$.

\subsection{Binomial Characterization}
\label{subsec:proof_binomial}

Theorem~\ref{thm:binomial_characterization} establishes a close connection between
the expected size $\Ev[\nwspace |\min(\B_{n,m,p})| \nwspace ]$ of the minimization
and the binomial distribution.
We split its proof into three lemmas corresponding, in that order,
to the lower bound in the first statement of the theorem,
the tight upper bound, and the second statement.

Note that the following lemma is slightly more general than what is claimed in 
Theorem~\ref{thm:binomial_characterization}~(\emph{\ref{clause:char_thm_almostall}}).
It pertains to all $m \le 1/(1-p)^n$,
equivalently, it does not require $\alpha$ to be bounded away from $1$.

\begin{lemma}[Lower bound of Theorem~\ref{thm:binomial_characterization}~(\emph{\ref{clause:char_thm_almostall}})]
\label{lem:char_almostall_lower}
	For all $m \le 1/(1-p)^{n}$,
	i.e., all $0 \le \alpha \le 1$, it holds that
	$\Ev[ \nwspace |\min(\B_{n,m,p})| \nwspace ] 
			= \Om(m) \cdot \Pb[ \nwspace X \le i^* \nwspace]$.
\end{lemma}

\begin{proof}
	The sought expectation is at least as large as
	the number of distinct minimal sets up to some cardinality $i$,
	for an arbitrary choice of $i \le n$.
	As an ansatz, we set it equal to the threshold
	$i^* = n + \log_{1-p} m = (1-\alpha) n$.
	Without loosing generality, $i^*$ is an integer;
	otherwise, we take $\lfloor i^* \rfloor$.
	Note that the assumption $m \le 1/(1-p)^n$ ensures $i^* \ge 0$.
	Let $\p = \max_{0 \le i \le i^*} \{ \nwspace p^i (1-p)^{n-i} \nwspace \}$.
	Lemmas~\ref{lem:formula_expt_Bnk}~($\ref{clause:lower_bound}$)
	and \ref{lem:distinct_sets} imply
	\begin{align*}
		\Ev[\nwspace |\min(\B_{n,m,p})| \nwspace]
			&\ge \sum_{i=0}^{i^*} \binom{n}{i} \, s_{n,p}(i,m) \cdot w_{n,p}(i,m)
			 \ge \sum_{i=0}^{i^*} \binom{n}{i} \, s_{n,p}(i,m) \cdot w_{n,p}(i^*,m)\\
			&= \Ev[\nwspace \|\B_{n,m,p} (0,i^*)\| \nwspace] \cdot w_{n,p}(i^*,m)
			 \ge \frac{1}{1+ m \p} \cdot w_{n,p}(i^*,m)  
				\cdot m \Pb[\nwspace X \le i^* \nwspace].
	\end{align*}
	
	To complete the proof, we verify that the first two factors do not vanish.
	The first factor is immediate from the closed form for $\p$
	(Lemma~\ref{lem:distinct_sets}) since
	$m\p = \max_{0 \le i \le i^*} \{ \nwspace m p^i (1-p)^{n-i} \nwspace \}
			\le \max \{1, p^{i^*}\} = 1$.
	Lemma~\ref{lem:threshold_fct}~($\ref{clause:exact_bounds}$)
	shows that there exists a universal constant $\delta > 0$ for all $m \le 1/(1-p)^n$
	such that $w(i^*,m) \ge \delta$.	
\end{proof}

\begin{lemma}[Upper bound of Theorem~\ref{thm:binomial_characterization}~(\emph{\ref{clause:char_thm_almostall}})]
\label{lem:char_almostall_upper}
	For any $\varepsilon > 0$ and all $m \le 1/(1-p)^{(1-\varepsilon) n}$, 
	i.e., all $0 \le \alpha \le 1 - \varepsilon$,
	$\Ev[ \nwspace |\min(\B_{n,m,p})| \nwspace ] 
	= \Or(m) \cdot \Pb[ \nwspace X \le i^* \nwspace]$.
	The leading constant may depend on $\varepsilon$.
\end{lemma}

\begin{proof}
	We know that
	$\Ev[\nwspace |\min(\B_{n,m,p})| \nwspace]
	\le \sum_{i=0}^n \binom{n}{i} s_{n,p}(i,m) \cdot w_{n,p}(i, m -1)$.
	We split the sum at the threshold $i^* = (1-\alpha)n$ and handle the two parts separately.
	The assumption on $m$ is such that $i^* \ge \varepsilon n$.
	For the first part,
	Lemma~\ref{lem:distinct_sets} implies
	\begin{equation*}
		\sum_{i = 0}^{i^*} \binom{n}{i} s_{n,p}(i,m) \cdot w_{n,p}(i,m -1)
			\le \Ev[ \nwspace \| \B_{n,m,p} (0,i^*) \| \nwspace ]
			\le m \cdot \Pb[\nwspace X \le i^* \nwspace].
	\end{equation*}
	
	For cardinalities larger than $i^*$,
	we can no longer ignore the influence of the weighting factors.
	We show that the whole second part of the sum is within constant factors
	of $m \cdot \Pb[\nwspace X = i^* \nwspace]$.
	Let $\ell \le n - i^*$ be a positive integer.
	Consider the $(i^*{+}\ell)$-th term.
	If $\ell = n - i^*$ this is the last one and contributes at most $1$. 
	Assume $\ell < n - i^*$.
	As in the proof of Lemma~\ref{lem:distinct_sets},
	we see that the term is upper-bounded by
	$m \cdot \Pb[\nwspace X = i^* + \ell \nwspace] \cdot w_{n,p}(i^* {+} \ell, m)$.
	Dividing by $m \cdot \Pb[\nwspace X = i^* \nwspace]$ gives
	\begin{align*}
		\frac{\binom{n}{i^*+\ell} \nwspace s_{n,p}(i^* + \ell,m) \cdot w_{n,p}(i^* + \ell,m -1)}
			{m \nwspace \Pb[\nwspace X = i^* \nwspace]}
			&\le \frac{m \Pb[\nwspace X = i^* + \ell \nwspace]}{m \Pb[\nwspace X = i^* \nwspace]}
				\cdot w_{n,p}(i^* + \ell,m -1)\\
			&= \frac{\binom{n}{i^*+\ell}}{\binom{n}{i^*}} \left(\frac{p}{1-p}\right)^\ell
				w_{n,p}(i^* + \ell,m -1).
	\end{align*}

	\noindent
	The first factor is bounded for any fixed $\ell$, using $\alpha \le 1 -\varepsilon$.
	\begin{equation*}
		\frac{\binom{n}{i^*+\ell}}{\binom{n}{i^*}}
			= \prod_{j=1}^\ell \frac{n - i^*- j+1}{i^* + j}
			\le \left(\frac{n - i^*}{i^*}\right)^\ell
			= \left(\frac{\alpha}{1-\alpha}\right)^\ell
			\le \frac{1}{\varepsilon^\ell}.
	\end{equation*}
	
	\noindent
	Applying Lemma~\ref{lem:threshold_fct}~($\ref{clause:exact_bounds}$)
	and the same idea as in Lemma~\ref{lem:formula_expt_Bnk}~($\ref{clause:constant_factor}$)
	to the weighting factor yields
	\begin{align*}
		w_{n,p}(i^*+\ell, m -1) 
			&= \frac{w_{n,p}(i^*+\ell, m-1)}{w_{n,p}(i^*+\ell, m)} \cdot w_{n,p}(i^*+\ell, m)\\
			&\le \frac{1}{p} \nwspace \exp(-m \nwspace (1-p)^{n - i^* -\ell+1})
			 = \frac{1}{p} \nwspace \exp(- (1-p)^{-\ell+1}).
	\end{align*}
	
	\noindent
	Here, we used $\ell < n - i^*$,
	ensuring that the ratio between subsequent factors is a constant.
	
	Define $a = p/\varepsilon (1-p)$ and $b = 1/(1-p)$.
	So far, we have established that the ratio between the term at $i^*+\ell$
	and $m \cdot \Pb[ \nwspace X = i^* \nwspace]$ is at most
	\begin{equation*}
		r(\ell) = \frac{1}{p} \nwspace \frac{a^\ell}{\exp(b^{\ell-1})}.
	\end{equation*}
	
	\noindent
	The function $r$ itself is free of any dependence on $n$, $m$, or $\alpha$,
	but in order to prove our claim, 
	we need to bound the sum $1 + \sum_{\ell = 1}^{n - i^*-1} r(\ell)$.
	We show that the series
	$\sum_{\ell = 1}^\infty r(\ell)$ is in fact summable.

	To this end,
	consider the sequence $t(\ell) = r(\ell) \cdot 2^\ell$.
	Its logarithm $\ln t  = \ell \nwspace \ln (2a) - \ln p - b^{\ell-1}$ 
	diverges to $- \infty$ as $\ell$ increases, implying $t \to 0$.
	This means, there exists an $\ell_0$ such that $r(\ell) \le 2^{-\ell}$ for all $\ell \ge \ell_0$.
	\begin{equation*}
		\sum_{\ell = 1}^\infty r(\ell)
%			= \sum_{\ell = 1}^\infty \frac{a^\ell}{p \cdot \exp(b^{\ell-1})}
			\le \sum_{\ell = 1}^{\ell_0} \frac{a^\ell}{p \cdot \exp(b^{\ell-1})} 
				+ \sum_{\ell = \ell_0 + 1}^\infty \frac{1}{2^\ell}
			\le \sum_{\ell = 1}^{\ell_0} \frac{a^\ell}{p \cdot \exp(b^{\ell-1})} + 2 = \Or(1). \qedhere
	\end{equation*}
\end{proof}

Finally, we show that once $m$ is a polynomial factor larger than $1/(1-p)^n$,
the minimization essentially consists of only a single edge, the empty set.

\begin{lemma}[Theorem~\ref{thm:binomial_characterization}~(\emph{\ref{clause:char_thm_trivial_case}})]
\label{lem:char_trivial}
	There exists a constant $c>0$, possibly dependent on $p$,
	such that for all $m \ge 1/(1-p)^{n+ c \ln n}$, it holds that
	$1 \le \Ev[ \nwspace |\min(\B_{n,m,p})| \nwspace ] = 1 + \littleO(1)$.
\end{lemma}

\begin{proof}
	The lower bound is immediate.
	Suppose we have $m = 1/(1-p)^{n + f}$ for some non-negative function $f = f(n)$.
	As soon as the empty set is sampled in one of the $m$ trials,
	the minimization of $\B_{n,m,p}$ comprises only a single set;
	otherwise, we fall back to the trivial estimate $|\min(\B_{n,m,p})| \le m$.
	Let $A$ denote the event $[\emptyset \in \B_{n,m,p}]$.
	The law of total expectation together with
	Corollary~\ref{cor:tight_bound_polynom} implies that
	\begin{align*}
		\Ev[\nwspace |\min(\B_{n,m,p})| \nwspace]
			&= \Ev[\nwspace |\min(\B_{n,m,p})| \mid A] \cdot \Pb[A] 
				+ \Ev[\nwspace |\min(\B_{n,m,p})| \mid \neg A] \cdot \Pb[\neg A]\\
			&\le \Pb[A] + m \cdot (1-(1-p)^{n})^{m}
			 \le 1 + \exp\!\left(\ln m - m \nwspace (1-p)^n \right)\\
			&= 1+ \exp\!\left(\ln m - (1-p)^{-f} \right).
	\end{align*}
	
	\noindent
	We have $\ln m = - \ln(1-p)(n + f)$, where $-\ln(1-p)$ is a constant strictly larger than $1$.
	Requiring $f \ge c' \nwspace (\ln n)/(- \ln(1-p))$ for some arbitrary $c' > 1$ ensures that
	$\ln m$ is negligible compared to $(1-p)^{-f}$ and the expression above converges to $1$.
	Setting $c = -c'/\ln(1-p)$ gives the lemma.
\end{proof}

\subsection{The Case $m = 1/(1-p)^n$}
\label{subsec:proof_special_case}

There is a mismatch in the upper and lower bounds above
in the range of parameters for which they hold.
The lower bound has been shown for the full range of \mbox{$m = 1/(1-p)^{\alpha n} \le 1/(1-p)^n$},
this includes the case where the function $\alpha(n)$ converges to $1$.
We get from Lemma~\ref{lem:char_almostall_lower} that
$\Ev[\nwspace |\min(\B_{n,m,p})| \nwspace] = \Om(1)$ 
if $\alpha = 1$ (for all $n$ sufficiently large).
This is not very surprising as already the the definition
of $\min(\B_{n,m,p})$ guarantees the existence of at least one minimal edge.
The upper bound in Lemma~\ref{lem:char_almostall_upper} is more interesting.
It only has been proven for $\alpha \le 1 - \varepsilon$ for any $\varepsilon > 0$.
If instead we were to insert, say, $\alpha = 1- \littleO(\sfrac{1}{n})$,
which corresponds to $m = 1/(1-p)^{n - \littleO(1)}$,
we would get $\Pb[\nwspace X \le (1-\alpha)n \nwspace] = \Pb[\nwspace X = 0 \nwspace] = (1-p)^n$
and the result would default to $\Or(1)$.
We refute this below.
Contrarily, we prove a lower bound for $\alpha = 1$ 
that is stronger than the inherited one.
This shows that the binomial characterization breaks down
if $\alpha$ converges to $1$.

\begin{lemma}
\label{lem:cannot_extend}
	For $m = 1/(1-p)^n$, i.e., $\alpha = 1$, it holds that 
	$\Ev[\nwspace |\min(\B_{n,m,p})| \nwspace] = \Om(n)$.
\end{lemma}

\begin{proof}
	We extend the idea of the proof of Lemma~\ref{lem:char_trivial}.
	If none of the $m$ trials produces the empty set,
	then all distinct sampled singletons are minimal.
	\begin{equation*}
		\Ev[\nwspace |\min(\B_{n,m,p})| \nwspace] 
			\ge \Ev[\nwspace \|\B_{n,m,p}(1,1)\| \nwspace] \cdot \Pb[ \emptyset \notin \B_{n,m,p}]
			= n \cdot s_{n,p}(1,m) \cdot (1- s_{n,p}(0,m)).
	\end{equation*}
	
	\noindent
	We show that the product $s_{n,p}(1,m) \nwspace (1- s_{n,p}(0,m))$
	is bounded away from $0$ for all $n$.
	From Corollary~\ref{cor:tight_bound_polynom}, Lemma~\ref{lem:Badkobeh_bound}, and the assumption $m = 1/(1-p)^n$,
	it follows that
	\begin{equation*}
		s_{n,p}(1,m)
			= 1-(1 - p(1-p)^{n-1})^m \ge \frac{m \nwspace p (1-p)^{n-1}}{1 + m \nwspace p (1-p)^{n-1}}
%			= \frac{p/(1-p)}{1 + p/(1-p)} 
			= p.
	\end{equation*}
	\noindent
	For the second factor, we have
	\begin{equation*}
		1- s_{n,p}(0,m) = (1 - (1-p)^{n})^m
			\ge \exp(- m \nwspace (1-p)^n) \cdot (1- m \nwspace (1-p)^{2n})
			= \frac{1- (1-p)^{n}}{e}.\ \ \qedhere
	\end{equation*}
\end{proof}

Lemmas~\ref{lem:char_trivial} and \ref{lem:cannot_extend} together 
show that a slight polynomial increase (in $n$) of the number of trials beyond $1/(1-p)^n$
is enough to push the size of the minimization from at least linear to $1$.
We leave it as an open problem to give exact bounds for the transitional behavior of
$|\min(\B_{n,m,p})|$ around $m = 1/(1-p)^n$, i.e., $\alpha = 1$.

\subsection{Phase Transition at $m = 1/(1-p)^{(1-p)n}$}
\label{subsec:proof_phase_transition}

We show that the size of the minimization undergoes a phase transition at $m^* = 1/(1-p)^{(1-p)n}$.
This is made explicit in Theorem~\ref{thm:properties_Bnpk} (restated below)
and illustrated in Figure~\ref{fig:plots}.
Intuitively, for a small number of edges, the ratio of minimal edges among all edges is constant
and thus $|\min(\B_{n,m,p})|$ scales linearly with $m$.
At the transition point, both the size of the multi-hypergraph as well as its minimization
are of order $m^* = 1/(1-p)^{(1-p)n} = 2^{(\Hb(1-p) + p \ld p) \nwspace n}$.
This overlap is indicated in Figure~\ref{fig:plots-by-alpha} by dashed lines.
Beyond that, in the information-theoretic regime,
the size of the minimization is instead given by the perplexity of $\alpha$,
it follows $2^{(\Hb(\alpha) + (1-\alpha) \ld p) \nwspace n}$ up to polynomial factors.
To gauge the transition more accurately,
we reuse the correction term $\varphi = \varphi(n,p,1-\alpha)$
from Section~\ref{sec:Chernoff--Hoeffding}.
Close to the transition, for $\alpha \approx 1-p$, we have $\varphi \approx 1$.
As $\alpha$ moves further away, $\varphi$ is shrinking
and once there is a constant additive gap between $\alpha$ and $1-p$
it is of order $\varphi = \Th(\sfrac{1}{\sqrt{n}})$.
In summary, our results for the information-theoretic regime show that
the growth of $|\min(\B_{n,m,p})|$ continues at first after the transition,
but is now only sublinear in $m$.
The minimization peaks at $m = 1/(1-p)^{\frac{n}{1+p}}$ (see Lemma~\ref{lem:maximum_Bnpk})
and for larger $m$, the number of minimal edges is even falling exponentially,
although the number of trials further increases.
Once $m$ exceeds $1/(1-p)^n$, the minimization collapses under the sheer likelihood
of the empty set being sampled.

\begin{backInTime}{thm:properties_Bnpk}
\begin{theorem}[restated]
	Let  $n$, $m$ be positive integers and $0 < p < 1$ a non-trivial probability.
	\begin{enumerate}[\textup{(}$i$\textup{)}]
		\item
			If $m \le 1/(1-p)^{(1-p)n}$,
			then $\Ev[ \nwspace |\min(\B_{n,m,p})| \nwspace] = \Th(m)$.\vspace*{.5em}
		\item
			For any $\varepsilon > 0$ and all 
			$1/(1-p)^{(1-p)n} \le m \le 1/(1-p)^{(1-\varepsilon) n}$,
			i.e., all $1-p \le \alpha \le 1 - \varepsilon$,
			\begin{align*}
				\Ev[ \nwspace |\min(\B_{n,m,p})| \nwspace ]
					&= \Th(1) \cdot \min\!\left(1, 
						\, \frac{1}{(\alpha\,{-}\,(1-p)) \nwspace \sqrt{(1\,{-}\,\alpha) n}}\right) 
	 					\cdot 2^{(\Hb(\alpha) + (1-\alpha) \log_2 p)  \nwspace n}\\[.25em]
					&= \Th(1) \cdot \min\!\left(1, 
						\, \frac{1}{(\alpha\,{-}\,(1-p)) \nwspace \sqrt{(1\,{-}\,\alpha) n}}\right) 
						\cdot \left(\frac{p^{1-\alpha}}
						{ (1-\alpha)^{1-\alpha} \ \alpha^\alpha} \right)^n\!.
			\end{align*}
	\end{enumerate}
\end{theorem}
\end{backInTime}

\begin{proof}
	The first statement covers the linear regime of $m \le 1/(1-p)^{(1-p)n}$.
	By the binomial characterization in Theorem~\ref{thm:binomial_characterization},
	we have $\Ev[\nwspace |\min(\B_{n,m,p})| \nwspace] = \Th(m) \cdot \Pb[X \le i^* \nwspace]$,
	with $X \sim \Bin(n,p)$.
	It is thus enough to verify that $\Pb[X \le i^*] = \Pb[X \le (1-\alpha)n \nwspace]$
	does not converge to $0$.
	This follows easily from $\alpha \le 1-p$ and $pn$ being the median of the binomial distribution.

	In the remainder, we treat the information-theoretic regime of all $m$ such that
	$1/(1-p)^{(1-p)} \le m \le 1/(1-p)^{(1-\varepsilon) n}$ for some fixed $\varepsilon > 0$.
	In particular, $(1-\alpha)n$ is now smaller than $\Ev[X]$.
	Recall the definition of function $\varphi$
	from Section~\ref{sec:Chernoff--Hoeffding},
	inserting $1-\alpha$ for $x$ gives
	\begin{equation*}
		\varphi = \varphi(n,p, 1\,{-}\,\alpha) = \min\!\left( 1, \,\frac{1}{(\alpha - (1-p)) \nwspace \sqrt{(1-\alpha) n}} \right)\!.
	\end{equation*}

	\noindent
	Our improved Chernoff--Hoeffding bound (Theorem~\ref{thm:better_Hoeffding}) implies
	$\Ev[\nwspace |\min(\B_{n,m,p})| \nwspace] 
		= \Th(m) \cdot \varphi \cdot 2^{-\!\Dkl(1-\alpha \,{\|}\, p) \nwspace n}$.
	Expressing the (the power of) the divergence in terms of the perplexity,
	and using $\Hb(1-\alpha) = \Hb(\alpha)$ as well as $m = 1/(1-p)^{\alpha n}$,
	shows that this is equal to
	\begin{align*}
		\Th(1) \cdot \varphi \cdot m \ 2^{-\!\Dkl(1-\alpha \,{\|}\, p) \nwspace n}
			&= \Th(1) \cdot \varphi \cdot \frac{1}{(1-p)^{\alpha n}} \ 2^{\Hb(1-\alpha)n}
				\ p^{(1-\alpha) n} \, (1-p)^{\alpha n}\\
			&= \Th(1) \cdot \varphi \cdot 2^{(\Hb(\alpha) + (1-\alpha) \ld p) \nwspace n}.\qedhere
	\end{align*}
\end{proof}

The behavior of the minimization for increasing $m$ suggests
that there is a sweet spot in the information-theoretic regime
where the expected number of minimal edges is maximum.
We apply Theorem~\ref{thm:properties_Bnpk} to calculate this maximum.
%In more detail, we aim to show a bound of $\Th((1+p)^n/\sqrt{n})$.

\begin{lemma}
\label{lem:maximum_Bnpk}
	The maximum expected number of minimal edges is
	$\max_{m \ge 1} \ \Ev[ \nwspace |\min(\B_{n,m,p})| \nwspace ]
		= \Th\!\left((1+p)^n/\sqrt{n}\right)$,
	it is attained for $m = 1/(1-p)^{\frac{n}{1+p}}$.\vspace*{.5em}
\end{lemma}

\begin{proof}
	We first verify that
	that the maximum indeed sits in the information-theoretic regime. 
	Observe that  $1/(1-p)^{1-p} < 1+p$ holds for all non-trivial probabilities $p$.
	This can be seen from $1/(1-p)^{1-p}$ being strictly concave
	on the open unit interval and $1+p$ being its tangent line at $p=0$.
	Theorem~\ref{thm:properties_Bnpk}~($\ref{clause:main_thm_linear_range}$) shows that
	the sample sizes in the linear regime are too small to lead to the claimed bound 
	of $\Th^*((1+p)^n)$.
	Moreover, $2^{\Hb(\alpha) + (1-\alpha) \ld p}$ as a function of $\alpha$ is continuous and concave,
	it converges to $2^{\Hb(1-p) + p \ld p} = 1/(1-p)^{(1-p)}$
	as $\alpha \searrow 1-p$ (for any $p$).
	Hence, there exists some $\delta > 0$ small enough such that
	$2^{\Hb(1-p + \delta) + (p-\delta) \ld p} < 1 + p$.
	Similarly, the function converges to $1$ from above as $\alpha \nearrow 1$.
	Let $\varepsilon > 0$ be such that
	$2^{\Hb(1-\varepsilon) + \varepsilon \ld p} < 1 + p$.
	In other words, any bound exponential in $1\,{+}\,p$ must come from an
	$\alpha  \in [1\,{-}\,p\,{+}\,\delta, 1\,{-}\,\varepsilon]$,
	we are in the setting of
	Theorem~\ref{thm:properties_Bnpk}~($\ref{clause:main_thm_information-theoretic_range}$).
	
	There are suitable constants $C,C' >0$ such that
	\begin{equation*}
		C \cdot \varphi \cdot 2^{(\Hb(\alpha) + (1-\alpha) \ld p) \nwspace n}
			\le \Ev[ \nwspace |\min(\B_{n,m,p})| \nwspace ]
			\le C' \cdot \varphi \cdot 2^{(\Hb(\alpha) + (1-\alpha) \ld p) \nwspace n}.
	\end{equation*}
	
	\noindent
	As $\alpha$ is bounded away from both $1-p$ and $1$, we have
	$\varphi = \Th(\sfrac{1}{\sqrt{n}})$
	and the hidden constants are independent of $\alpha$
	(but possibly depending on $\delta$ and $\varepsilon$).
	We are thus left to determine the extremum of the exponential part.
	Let 
	\begin{equation*}
		g(\alpha,p) 
			= \Hb(\alpha) + (1-\alpha) \nwspace \ld p 
			= -\alpha \ld \alpha - (1-\alpha) \ld (1-\alpha) + (1-\alpha) \ld p
	\end{equation*}
	be the exponent to base $2^n$ of the expression above.
%	As a sum of positive concave functions, $g$ is positive and concave as well.
	Its partial derivative 
 	\begin{equation*}
	 	\frac{\partial}{\partial \alpha} \nwspace g(\alpha,p)
	 		= \ld\!\left( \frac{1-\alpha}{\alpha} \right) - \ld p
	 \end{equation*} 
	
	\noindent
	has a single zero in the interval $[1-p+\delta, 1-\varepsilon]$ at $\alpha^* = 1/(1+p)$,
	resulting in an exponent of
	\begin{equation*}
		g(\alpha^*\!, p)
			= - \frac{1}{1+p} \nwspace \ld\!\left(\frac{1}{1+p}\right) 
				- \frac{p}{1+p} \nwspace \ld\!\left(\frac{p}{1+p}\right)
				+ \frac{p}{1+p} \nwspace \ld p
			= \ld(1+p). \qedhere
	\end{equation*}
\end{proof}

\section{Conclusion}
\label{sec:conclusion}

We calculated the number of minimal edges of maximum-entropy multi-hypergraphs 
with expected edge size $pn$
and discovered a phase transition with respect to the total number of edges.
As long as $m$ is at most $m^* = 1 /(1 - p)^{(1 - p)n}$,
the expected size of the minimization is linear in $m$.
Once the sample size increases beyond that point,
the minimization is instead governed by the entropy of $\alpha$
where $m = 1/(1-p)^{\alpha n}$.
The minimization continues to grow sublinearly 
until $m$ reaches $1/(1 - p)^ {\frac{n}{1+p}}$,
from there on, it decays rapidly. 
Raising $m$ above $1/(1 - p)^{n}$
finally ensures that only the empty set is minimal.

The Chernoff--Hoeffding theorem played an integral role in verifying these results.
We tightened the tail bounds on the binomial distribution 
and provided explicit upper and lower bounds on the constants involved.
We are convinced that this improvement can help other researchers in probability
well beyond the scope of this paper.

A possible extension of our work would be to generalize Theorem~\ref{thm:binomial_characterization}
to the case where $\alpha$ converges to $1$ as $n$ increases.
It is also interesting to investigate hypergraph models 
for which the sampling probability $p$ is no longer constant, but tends to $0$.
This would bring the study of random hypergraphs closer to that of sparse random graphs.
Finally, a promising direction in the light of our original motivation of relational databases
is to allow different sample probabilities per vertex as well as
dependencies between the elements.
This requires the maximum-entropy model to incorporate additional constraints.

\subsection*{Acknowledgments}

The authors thank Benjamin Doerr, Timo Kötzing, and Martin Krejca
for the fruitful discussions on the Chernoff--Hoeffding theorem,
including valuable pointers to the literature.

\bibliography{../literature}

\begin{thebibliography}{66}
\providecommand{\natexlab}[1]{#1}
\providecommand{\url}[1]{\texttt{#1}}
\expandafter\ifx\csname urlstyle\endcsname\relax
  \providecommand{\doi}[1]{doi: #1}\else
  \providecommand{\doi}{doi: \begingroup \urlstyle{rm}\Url}\fi

\bibitem[Abedjan et~al.(2018)Abedjan, Golab, Naumann, and
  Papenbrock]{Abedjan2018DataProfilingBook}
Z.~Abedjan, L.~Golab, F.~Naumann, and T.~Papenbrock.
\newblock \emph{{Data Profiling}}.
\newblock Synthesis Lectures on Data Management. Morgan {\&} Claypool
  Publishers, San Rafael, CA, USA, 2018.

\bibitem[Anand and Bianconi(2009)]{Anand09Entropy}
K.~Anand and G.~Bianconi.
\newblock {Entropy Measures for Networks: Toward an Information Theory of
  Complex Topologies}.
\newblock \emph{Physical Review E}, 80:\penalty0 045102, 2009.

\bibitem[Arafat et~al.(2020)Arafat, Basu, Decreusefond, and
  Bressan]{Arafat20HypergraphsPrescribedDegreeSequenceArXiv}
N.~A. Arafat, D.~Basu, L.~Decreusefond, and S.~Bressan.
\newblock {Construction and Random Generation of Hypergraphs with Prescribed
  Degree and Dimension Sequences}.
\newblock \emph{CoRR}, abs/2004.05429, 2020.
\newblock ArXiv preprint, to appear at DEXA 2020.

\bibitem[Ash(1990)]{Ash90InformationTheory}
R.~B. Ash.
\newblock \emph{Information Theory}.
\newblock Dover Books on Mathematics. Dover Publications, Mineola, NY, USA,
  1990.
\newblock Reprint of the Interscience Publishers 1965 edition.

\bibitem[Badkobeh et~al.(2015)Badkobeh, Lehre, and Sudholt]{Badkobeh:2015}
G.~Badkobeh, P.~K. Lehre, and D.~Sudholt.
\newblock {Black-box Complexity of Parallel Search with Distributed
  Populations}.
\newblock In \emph{Proceedings of the 2015 Conference on Foundations of Genetic
  Algorithms (FOGA)}, pp. 3--15, 2015.

\bibitem[Ball et~al.(2017)Ball, Rosen, Sabin, and
  Vasudevan]{Ball17AverageCaseFineGrainedHardness}
M.~Ball, A.~Rosen, M.~Sabin, and P.~N. Vasudevan.
\newblock {Average-case fine-grained hardness}.
\newblock In H.~Hatami, P.~McKenzie, and V.~King, editors, \emph{Proceedings of
  the 49th Annual Symposium on Theory of Computing, (STOC)}, pp. 483--496,
  2017.

\bibitem[Behrisch et~al.(2010)Behrisch, Coja-Oghlan, and
  Kang]{Behrisch010GiantComponent}
M.~Behrisch, A.~Coja-Oghlan, and M.~Kang.
\newblock {The Order of the Giant Component of Random Hypergraphs}.
\newblock \emph{Random Structures and Algorithms}, 36:\penalty0 149--184, 2010.

\bibitem[Behrisch et~al.(2014)Behrisch, Coja-Oghlan, and
  Kang]{Behrisch14LocalLimit}
M.~Behrisch, A.~Coja-Oghlan, and M.~Kang.
\newblock {Local Limit Theorems for the Giant Component of Random Hypergraphs}.
\newblock \emph{Combinatorics, Probability and Computing}, 23:\penalty0
  331--–366, 2014.

\bibitem[Berge(1989)]{BergeHypergraphs}
C.~Berge.
\newblock \emph{{Hypergraphs - Combinatorics of Finite Sets}}, Vol.~45 of
  \emph{North-Holland Mathematical Library}.
\newblock North-Holland, Amsterdam, Netherlands, 1989.

\bibitem[Bertossi(2011)]{Bertossi11DatabaseRepairingBook}
L.~Bertossi.
\newblock \emph{Database Repairing and Consistent Query Answering}.
\newblock Synthesis Lectures on Data Management. Morgan {\&} Claypool
  Publishers, San Rafael, CA, USA, 2011.

\bibitem[Bianconi(2007)]{Bianconi07NetworkEnsembles}
G.~Bianconi.
\newblock {The Entropy of Randomized Network Ensembles}.
\newblock \emph{Europhysics Letters}, 81:\penalty0 28005, 2007.

\bibitem[Bl{\"a}sius et~al.(2016)Bl{\"a}sius, Friedrich, and
  Schirneck]{Blaesius16IPEC}
T.~Bl{\"a}sius, T.~Friedrich, and M.~Schirneck.
\newblock {The Parameterized Complexity of Dependency Detection in Relational
  Databases}.
\newblock In \emph{Proceedings of the 11th International Symposium on
  Parameterized and Exact Computation (IPEC)}, pp. 6:1--6:13, 2016.

\bibitem[Bl{\"a}sius et~al.(2020)Bl{\"a}sius, Friedrich, and
  Schirneck]{Blaesius20MinimizationESA}
T.~Bl{\"a}sius, T.~Friedrich, and M.~Schirneck.
\newblock {The Minimization of Random Hypergraphs}.
\newblock In \emph{Proceedings of the 28th Annual European Symposium on
  Algorithms (ESA)}, pp. 21:1--21:15, 2020.

\bibitem[Bleifu{\ss} et~al.(2017)Bleifu{\ss}, Kruse, and
  Naumann]{Bleifuss17DenialConstraints}
T.~Bleifu{\ss}, S.~Kruse, and F.~Naumann.
\newblock {Efficient Denial Constraint Discovery with Hydra}.
\newblock \emph{Proceedings of the VLDB Endowment}, 11:\penalty0 311–323,
  2017.

\bibitem[Bläsius et~al.(2019)Bläsius, Friedrich, Lischeid, Meeks, and
  Schirneck]{Blaesius2019Efficiently}
T.~Bläsius, T.~Friedrich, J.~Lischeid, K.~Meeks, and M.~Schirneck.
\newblock {Efficiently Enumerating Hitting Sets of Hypergraphs Arising in Data
  Profiling}.
\newblock In \emph{Proceedings of the 21st Meeting on Algorithm Engineering and
  Experiments (ALENEX)}, pp. 130--143, 2019.

\bibitem[Bollob{\'a}s(2001)]{Bollobas01RandomGraphs}
B.~Bollob{\'a}s.
\newblock \emph{{Random Graphs}}.
\newblock Studies in Advanced Mathematics. Cambridge University Press,
  Cambridge, UK, 2 edition, 2001.

\bibitem[Bollobás(1980)]{Bollobas80ConfigurationModel}
B.~Bollobás.
\newblock {A Probabilistic Proof of an Asymptotic Formula for the Number of
  Labelled Regular Graphs}.
\newblock \emph{European Journal of Combinatorics}, 1:\penalty0 311--316, 1980.

\bibitem[Borassi et~al.(2016)Borassi, Crescenzi, and
  Habib]{Borassi2016IntoTheSquare}
M.~Borassi, P.~Crescenzi, and M.~Habib.
\newblock {Into the Square: On the Complexity of Some Quadratic-time Solvable
  Problems}.
\newblock \emph{Electronic Notes in Theoretical Computer Science},
  322:\penalty0 51--67, 2016.

\bibitem[Bruen and Forcinito(2004)]{Bruen04Cryptography}
A.~A. Bruen and M.~A. Forcinito.
\newblock \emph{{Cryptography, Information Theory, and Error-Correction: A
  Handbook for the 21st Century}}.
\newblock Wiley-Interscience, New York, NY, USA, 2004.

\bibitem[Chodrow(2020)]{Chodrow20ConfigurationComplexNetworks}
P.~S. Chodrow.
\newblock {Configuration Models of Random Hypergraphs}.
\newblock \emph{Journal of Complex Networks}, 8, 2020.
\newblock Article number cnaa018.

\bibitem[Cooper et~al.(2019)Cooper, Frieze, and Pegden]{Cooper19RankJournal}
C.~Cooper, A.~M. Frieze, and W.~Pegden.
\newblock {On the Rank of a Random Binary Matrix}.
\newblock \emph{Electronic Journal of Combinatorics}, 26, 2019.
\newblock Article number P4.12.

\bibitem[Cover and Thomas(2006)]{Cover2006InformationTheory}
T.~M. Cover and J.~A. Thomas.
\newblock \emph{{Elements of Information Theory}}.
\newblock Wiley Series in Telecommunications and Signal Processing.
  Wiley-Interscience, New York, NY, USA, 2nd edition, 2006.

\bibitem[Cram{\'e}r(1938)]{Cramer38SurUn}
H.~Cram{\'e}r.
\newblock Sur un nouveau th{\'e}or{\`e}me-limite de la th{\'e}orie des
  probabilit{\'e}s.
\newblock In \emph{Actualit{\'e}s scientifiques et industrielles}, Vol. 763,
  pp. 5--23, 1938.
\newblock Colloque consacr{\'e} {\`a} la th{\'e}orie des probabilit{\'e}s. (On
  a New Limit Theorem in Probability.) In French.

\bibitem[Dalirrooyfard et~al.(2020)Dalirrooyfard, Lincoln, and
  Vassilevska~Williams]{Dalirrooyfard20FineGrainedAverageCaseArXiv}
M.~Dalirrooyfard, A.~Lincoln, and V.~Vassilevska~Williams.
\newblock {New Techniques for Proving Fine-Grained Average-Case Hardness}.
\newblock \emph{CoRR}, abs/2008.06591, 2020.
\newblock ArXiv preprint, to appear at FOCS 2020.

\bibitem[Dembo and Zeitouni(2010)]{Dembo10LargeDeviations}
A.~Dembo and O.~Zeitouni.
\newblock \emph{{Large Deviations Techniques and Applications}}, Vol.~38 of
  \emph{Stochastic Modelling and Applied Probability}.
\newblock Springer, Berlin and Heidelberg, Germany, 2nd edition, 2010.

\bibitem[Demetrovics et~al.(1998)Demetrovics, Katona, Miklos, Seleznjev, and
  Thalheim]{Demetrovics98RandomDatabases}
J.~Demetrovics, G.~O.~H. Katona, D.~Miklos, O.~Seleznjev, and B.~Thalheim.
\newblock {Asymptotic Properties of Keys and Functional Dependencies in Random
  Databases}.
\newblock \emph{Theoretical Computer Science}, 190:\penalty0 151--166, 1998.

\bibitem[Deuschel and Stroock(1989)]{Deuschel84LargeDeviations}
J.-D. Deuschel and D.~W. Stroock.
\newblock \emph{{Large Deviations}}, Vol. 137 of \emph{Pure and Applied
  Mathematics}.
\newblock Academic Press, Boston, MA, USA, 1989.

\bibitem[Doerr(2020)]{Doerr20Evolutionary}
B.~Doerr.
\newblock {Probabilistic Tools for the Analysis of Randomized Optimization
  Heuristics}.
\newblock In B.~Doerr and F.~Neumann, editors, \emph{{Theory of Evolutionary
  Computation}}, chapter~1. Springer International, Basel, Switzerland, 2020.

\bibitem[Dubhashi and Panconesi(2009)]{Dubhashi09ConcentrationOfMeasure}
D.~Dubhashi and A.~Panconesi.
\newblock \emph{Concentration of Measure for the Analysis of Randomized
  Algorithms}.
\newblock Cambridge University Press, New York, NY, USA, 2009.

\bibitem[Erd{\H o}s and R{\'e}nyi(1959)]{Erdos59RandomGraphs}
P.~Erd{\H o}s and A.~R{\'e}nyi.
\newblock {On Random Graphs I}.
\newblock \emph{Publicationes Mathematicae Debrecen}, 6:\penalty0 290--297,
  1959.

\bibitem[Fill(1983)]{Fill83ConvrgenceRates}
J.~A. Fill.
\newblock {Convergence Rates Related to the Strong Law of Large Numbers}.
\newblock \emph{Annals of Probability}, 11:\penalty0 123--142, 1983.

\bibitem[Froese et~al.(2016)Froese, van Bevern, Niedermeier, and
  Sorge]{Froese16DistinctVectors}
V.~Froese, R.~van Bevern, R.~Niedermeier, and M.~Sorge.
\newblock {Exploiting Hidden Structure in Selecting Dimensions That Distinguish
  Vectors}.
\newblock \emph{Journal of Computer and System Sciences}, 82:\penalty0
  521--535, 2016.

\bibitem[Gao et~al.(2018)Gao, Impagliazzo, Kolokolova, and
  Williams]{Gao2018CompletenessFO}
J.~Gao, R.~Impagliazzo, A.~Kolokolova, and R.~R. Williams.
\newblock {Completeness for First-Order Properties on Sparse Structures With
  Algorithmic Applications}.
\newblock \emph{Transactions on Algorithms}, 15:\penalty0 23:1--23:35, 2018.

\bibitem[Garlaschelli and Loffredo(2008)]{Garlaschelli08UnbiasedInformation}
D.~Garlaschelli and M.~I. Loffredo.
\newblock {Maximum Likelihood: Extracting Unbiased Information from Complex
  Networks}.
\newblock \emph{Physical Review E}, 78:\penalty0 015101, 2008.

\bibitem[Gilbert(1959)]{Gilbert59RandomGraphs}
E.~N. Gilbert.
\newblock {Random Graphs}.
\newblock \emph{Annals of Mathematical Statistics}, 30:\penalty0 1141--1144,
  1959.

\bibitem[Grandvalet and Bengio(2006)]{Grandvalet06EntropyRegularization}
Y.~Grandvalet and Y.~Bengio.
\newblock {Entropy Regularization}.
\newblock In O.~Chapelle, B.~Sch{\"o}lkopf, and A.~Zien, editors,
  \emph{{Semi-Supervised Learning}}, chapter~9. MIT Press, Cambridge, MA, USA,
  2006.

\bibitem[Harremo{\"e}s(2001)]{Harremoes01MaximumEntropy}
P.~Harremo{\"e}s.
\newblock {Binomial and Poisson Distributions as Maximum Entropy
  Distributions}.
\newblock \emph{Transactions on Information Theory}, 47:\penalty0 2039--2041,
  2001.

\bibitem[Hoeffding(1963)]{Hoeffding63Inequalities}
W.~Hoeffding.
\newblock {Probability Inequalities for Sums of Bounded Random Variables}.
\newblock \emph{Journal of the American Statistical Association}, 58:\penalty0
  13--30, 1963.

\bibitem[Hofstad(2016)]{Hofstad16RandomGraphs}
R.~v.~d. Hofstad.
\newblock \emph{Random Graphs and Complex Networks}, Vol.~1 of \emph{Cambridge
  Series in Statistical and Probabilistic Mathematics}.
\newblock Cambridge University Press, Cambridge, UK, 2016.

\bibitem[Jaynes(1957{\natexlab{a}})]{Jaynes57MaximumEntropyI}
E.~T. Jaynes.
\newblock {Information Theory and Statistical Mechanics}.
\newblock \emph{Phyical Review Series II}, 106:\penalty0 620--630,
  1957{\natexlab{a}}.

\bibitem[Jaynes(1957{\natexlab{b}})]{Jaynes57MaximumEntropyII}
E.~T. Jaynes.
\newblock {Information Theory and Statistical Mechanics II}.
\newblock \emph{Phyical Review Series II}, 108:\penalty0 171--190,
  1957{\natexlab{b}}.

\bibitem[Kane and Williams(2018)]{KaneWilliams18OVBranchingPrograms}
D.~M. Kane and R.~R. Williams.
\newblock {The Orthogonal Vectors Conjecture for Branching Programs and
  Formulas}.
\newblock In \emph{Proceedings of the 10th Innovations in Theoretical Computer
  Science Conference (ITCS)}, pp. 48:1--48:15, 2018.

\bibitem[Karoński and Łuczak(2002)]{Karonski02PhaseTransition}
M.~Karoński and T.~Łuczak.
\newblock {The Phase Transition in a Random Hypergraph}.
\newblock \emph{Journal of Computational and Applied Mathematics},
  142:\penalty0 125--135, 2002.

\bibitem[Katona(2013)]{Katona13FunctionalConnection}
G.~O.~H. Katona.
\newblock {Testing Functional Connection Between Two Random Variables}.
\newblock In A.~N. Shiryaev, S.~R.~S. Varadhan, and E.~L. Presman, editors,
  \emph{Prokhorov and Contemporary Probability Theory}, pp. 335--348. Springer,
  Berlin and Heidelberg, Germany, 2013.
\newblock Festschrift.

\bibitem[Katona(2012)]{Katona2012CorrelatedData}
G.~O.~H. Katona.
\newblock {Random Databases with Correlated Data}.
\newblock In A.~D{\"u}sterh{\"o}ft, M.~Klettke, and K.-D. Schewe, editors,
  \emph{Conceptual Modelling and Its Theoretical Foundations: Essays Dedicated
  to Bernhard Thalheim on the Occasion of His 60th Birthday}, pp. 29--35.
  Springer, Berlin and Heidelberg, Germany, 2012.
\newblock Festschrift.

\bibitem[Kearns and Vazirani(1994)]{Kearns94ComputationalLearning}
M.~J. Kearns and U.~V. Vazirani.
\newblock \emph{{An Introduction to Computational Learning Theory}}.
\newblock MIT Press, Cambridge, MA, USA, 1994.

\bibitem[Kesavan(2009)]{Kesavan09Jaynes}
H.~K. Kesavan.
\newblock {Jaynes' Maximum Entropy Principle}.
\newblock In C.~A. Floudas and P.~M. Pardalos, editors, \emph{{Encyclopedia of
  Optimization}}, pp. 1779--1782. Springer, Boston, MA, USA, 2009.

\bibitem[Klar(2000)]{Klar00TailProbabilities}
B.~Klar.
\newblock {Bounds on Tail Probabilities of Discrete Distributions}.
\newblock \emph{Probability in the Engineering and Informational Sciences},
  14:\penalty0 161--171, 2000.

\bibitem[Lieb and Yngvason(1999)]{Lieb91PhysicsReport}
E.~H. Lieb and J.~Yngvason.
\newblock {The Physics and Mathematics of the Second Law of Thermodynamics}.
\newblock \emph{Physics Reports}, 310:\penalty0 1--96, 1999.

\bibitem[Livshits et~al.(2020)Livshits, Heidari, Ilyas, and
  Kimelfeld]{Livshits20ApproximatePVLDB}
E.~Livshits, A.~Heidari, I.~F. Ilyas, and B.~Kimelfeld.
\newblock {Approximate Denial Constraints}.
\newblock \emph{Proceedings of the VLDB Endowment}, 13:\penalty0 1682–1695,
  2020.

\bibitem[Mitzenmacher and Upfal(2017)]{Mitzenmacher17Probablility}
M.~Mitzenmacher and E.~Upfal.
\newblock \emph{{Probability and Computing}}.
\newblock Cambridge University Press, New York, NY, USA, 2nd edition, 2017.

\bibitem[Motwani and Raghavan(1995)]{Motwani1995RandAlg}
R.~Motwani and P.~Raghavan.
\newblock \emph{{Randomized Algorithms}}.
\newblock Cambridge University Press, Cambridge, UK, 1995.

\bibitem[Newman(2001)]{Newman01Collaboration}
M.~E.~J. Newman.
\newblock {Scientific Collaboration Networks. I. Network Construction and
  Fundamental Results}.
\newblock \emph{Physical Review E}, 64:\penalty0 016131, 2001.

\bibitem[Newman(2010)]{Newman10Networks}
M.~E.~J. Newman.
\newblock \emph{{Networks: An Introduction}}.
\newblock Oxford University Press, New York, NY, USA, 2010.

\bibitem[Nielsen and Chuang(2010)]{Nielsen10QuantumComputing}
M.~A. Nielsen and I.~L. Chuang.
\newblock \emph{{Quantum Computation and Quantum Information: 10th Anniversary
  Edition}}.
\newblock Cambridge University Press, Cambridge, UK, 2010.

\bibitem[O'Donnell(2014)]{ODonnell14BooleanFunctions}
R.~O'Donnell.
\newblock \emph{{Analysis of Boolean Functions}}.
\newblock Cambridge University Press, Cambridge, UK, 2014.

\bibitem[Oliveto and Witt(2015)]{Oliveto15ImprovedTimeComplexity}
P.~S. Oliveto and C.~Witt.
\newblock {Improved Time Complexity Analysis of the Simple Genetic Algorithm}.
\newblock \emph{Theoretical Computer Science}, 605:\penalty0 21 -- 41, 2015.

\bibitem[Papenbrock et~al.(2015)Papenbrock, Ehrlich, Marten, Neubert, Rudolph,
  Sch\"{o}nberg, Zwiener, and Naumann]{Papenbrock15SevenAlgorithms}
T.~Papenbrock, J.~Ehrlich, J.~Marten, T.~Neubert, J.-P. Rudolph,
  M.~Sch\"{o}nberg, J.~Zwiener, and F.~Naumann.
\newblock {Functional Dependency Discovery: An Experimental Evaluation of Seven
  Algorithms}.
\newblock \emph{Proceedings of the VLDB Endowment}, 8:\penalty0 1082--1093,
  2015.

\bibitem[Park and Newman(2004)]{Park04StatisticalMechanics}
J.~Park and M.~E.~J. Newman.
\newblock {The Statistical Mechanics of Networks}.
\newblock \emph{Physical Review E}, 70:\penalty0 066117, 2004.

\bibitem[Prokhorov(1953)]{Prokhorov53Asymptotic}
Y.~V. Prokhorov.
\newblock \foreignlanguage{russian}{Асимптотическое
  поведение биномиального распределения}.
\newblock \emph{\foreignlanguage{russian}{Успехи
  математических наук}}, 8:\penalty0 135--142, 1953.
\newblock (Asymptotic Behavior of the Binomial Distribution.) In Russian.

\bibitem[Saracco et~al.(2015)Saracco, Di~Clemente, Gabrielli, and
  Squartini]{Saracco15Bipartite}
F.~Saracco, R.~Di~Clemente, A.~Gabrielli, and T.~Squartini.
\newblock {Randomizing Bipartite Networks: The Case of the World Trade Web}.
\newblock \emph{Scientific Reports}, 5:\penalty0 10595, 2015.

\bibitem[Schmidt-Pruzan and Shamir(1985)]{SchmidtPruzan85Component}
J.~Schmidt-Pruzan and E.~Shamir.
\newblock {Component Structure in the Evolution of Random Hypergraphs}.
\newblock \emph{Combinatorica}, 5:\penalty0 81--94, 1985.

\bibitem[Shannon(1948)]{Shannon48Communication}
C.~E. Shannon.
\newblock {A Mathematical Theory of Communication}.
\newblock \emph{The Bell System Technical Journal}, 27:\penalty0 379--423,
  1948.

\bibitem[Slud(1977)]{Slud77BinomialLaw}
E.~V. Slud.
\newblock {Distribution Inequalities for the Binomial Law}.
\newblock \emph{The Annals of Probability}, 5:\penalty0 404--412, 1977.

\bibitem[Williams(2005)]{Williams2005SETHtoOV}
R.~R. Williams.
\newblock {A New Algorithm for Optimal 2-Constraint Satisfaction and Its
  Implications}.
\newblock \emph{Theoretical Computer Science}, 348:\penalty0 357--365, 2005.

\bibitem[Zweig(2014)]{Zweig14Literacy}
K.~A. Zweig.
\newblock \emph{{Network Analysis Literacy: A Practical Approach to the
  Analysis of Networks}}.
\newblock Lecture Notes in Social Networks. Springer, Vienna, Austria, 2014.

\end{thebibliography}
\end{document}